\documentclass[12pt, hidelinks]{article}

\usepackage[small, bf]{caption}

\usepackage[utf8]{inputenc}
\usepackage{fullpage}
\usepackage{appendix}
\usepackage{graphicx}
\usepackage{amsmath}
\usepackage{amssymb}
\usepackage{algorithm2e}
\usepackage{ifthen}
\usepackage{amsthm, amsfonts}
\usepackage{url}
\usepackage{mathtools}
\usepackage{chngcntr}
\usepackage{nicefrac}
\usepackage{hyperref}
\usepackage{subdepth}
\usepackage{forest}
\usepackage{tikzducks}
\usepackage{etoolbox}
\usepackage{hhline}
\usetikzlibrary{shapes.geometric,positioning}
\usetikzlibrary{decorations.pathmorphing,shapes}
\tikzset{elli/.style={ellipse,draw}}
\newsavebox\Car
\newsavebox\Tree

\usepackage{comment}
\usepackage{bbm}

\usepackage[style=alphabetic,backend=bibtex]{biblatex}
\bibliography{bibliography}

\theoremstyle{definition}

\theoremstyle{theorem}

\newtheorem{theorem}{Theorem}

\newtheorem{lemma}{Lemma}
\counterwithin*{lemma}{subsection}

\newtheorem{claim}{Claim}

\newcommand{\otheta}{\overline{\theta}}
\newcommand{\utheta}{\underline{\theta}}

\usepackage{tikz}
\usepackage{subcaption}
\usetikzlibrary{arrows,automata,calc}

\tikzset{
  treenode/.style = {align=center, inner sep=0pt, text centered,
    font=\sffamily},
  arn_n/.style = {treenode, circle, white, font=\sffamily\bfseries, draw=black,
    fill=black, text width=1.5em},
  arn_r/.style = {treenode, circle, red, draw=red, 
    text width=1.5em, very thick},
  arn_x/.style = {treenode, rectangle, draw=black,
    minimum width=0.5em, minimum height=0.5em}
}

\newcommand{\ones}{\mathbf 1}
\newcommand{\reals}{{\mbox{\bf R}}}

\newcommand{\naturals}{{\mbox{\bf N}}}

\newcommand{\dist}{\mathop{\bf dist{}}}

\newcommand{\argmax}{\mathop{\rm argmax}}

\newcommand{\eg}{{\it e.g.}}
\newcommand{\ie}{{\it i.e.}}

\newcommand{\BEAS}{\begin{eqnarray*}}
\newcommand{\EEAS}{\end{eqnarray*}}
\newcommand{\BEA}{\begin{eqnarray}}
\newcommand{\EEA}{\end{eqnarray}}
\newcommand{\BEQ}{\begin{equation}}
\newcommand{\EEQ}{\end{equation}}
\newcommand{\BIT}{\begin{itemize}}
\newcommand{\EIT}{\end{itemize}}

\title{How much should you pay for restaking security?}
\author{Tarun Chitra\\ Gauntlet \\\texttt{\small tarun@gauntlet.xyz} \and Mallesh Pai\\ Rice University and SMG\\\texttt{\small mallesh.pai@mechanism.org}}

\begin{document}
\maketitle

\begin{abstract}
Restaking protocols have aggregated billions of dollars of security by utilizing token incentives and payments.
A natural question to ask is: How much security do restaked services \emph{really} need to purchase?
To answer this question, we expand a model of Durvasula and Roughgarden~\cite{naveen-placeholder} that includes incentives and an expanded threat model consisting of strategic attackers and users.
Our model shows that an adversary with a strictly submodular profit combined with strategic node operators who respond to incentives can avoid the large-scale cascading failures of~\cite{naveen-placeholder}.
We utilize our model to construct an approximation algorithm for choosing token-based incentives that achieve a given security level against adversaries who are bounded in the number of services they can simultaneously attack.
Our results suggest that incentivized restaking protocols can be secure with proper incentive management. 
\end{abstract}

\section{Introduction}
Decentralized networks, such as blockchains, rely on a combination of cryptography and economic incentives to corral disparate operators to maintain network services.
These operators, often referred to as node operators, bear the cost of running infrastructure and maintaining high network availability to guarantee network service level agreements.
In exchange, these operators receive a combination of a fixed subsidy (often termed a block reward) and a variable fee that is accrued based on network usage.

In these systems, properties such as safety (\ie~valid transactions cannot be removed from the network) and liveness (\ie~the network does not halt) are dependent on how many resources are committed by network participants.
For instance, Byzantine Fault Tolerant (BFT) Proof of Stake (PoS) networks guarantee safety and liveness only if at most $\nicefrac{1}{3}$ of resources committed by node operators are adversarial and/or dishonest.
As such, these networks need to continually pay fixed and variable payments to ensure that there is sufficient honest stake.
As protocols often compete with applications built on top of blockchains (such as decentralized finance or DeFi) for stake, there is a minimum payment needed to achieve security (see, \eg~\cite{Chitra2021Competitive, chitra2022improving, chitra2020stake}).
This has led to a state where blockchains are continually searching for new forms of yield to give to node operators in order to ensure secure, orderly operation. 

\paragraph{Restaking.}
As blockchains have evolved, there have been various forms of fees paid to node operators.
The majority of fees earned by node operators are transaction fees that users pay in order to have their transactions included within a block.
These fees are generally fixed and not proportional to the value of the transactions involved. 
Other fee mechanisms unique to blockchains such as miner extractable value (MEV)~\cite{chitra2022improving, kulkarni2023routing, daian2020flash} also exist.
These forms of fees involve strategic users taking advantage of non-strategic user transactions and generally provide additional yet highly variable yield.

A newer form of yield for node operators is \emph{restaking}, pioneered by Eigenlayer~\cite{eigen-wp, eigen-token}. Restaking involves node operators of an existing PoS network (referred to as the \emph{host network}) locking their stake into a smart contract.
The node operators provide services to the smart contract that are in addition to the services the host network requires (\eg~ensuring transaction validity, voting on block finalization, etc.).
If the node operator does not meet a covenant, their stake that is locked in the contract on the host network is slashed.
On the other hand, if the node operator provides services within a service-level agreement (SLA) defined by the contract, they receive incentive payments.
To demonstrate the diversity of restaking services, note that currently live Eigenlayer actively validated services for price oracles, Zero Knowledge proof generation, rollup sequencing, decentralized exchanges, and AI co-processing~\cite{u1-restaking}.

To illustrate payments and slashing, consider a price oracle service that requires node operators to provide a price from an off-chain venue (\ie~Coinbase) on every block.
If the node operator provides a price, they receive a portion of the revenue that price oracle smart contract receives.
On the other hand, if the node operator doesn't tender a price, then the operator can have their stake slashed.
Note that this slashing rule is in addition to existing host network slashing rules (\ie~an Ethereum or Solana node operator is slashed for not posting a block during a slot they are the proposer).
This shows that restaking can be thought of as a node operator earning extra fees by opting into excess risk from a service's slashing rules.

\paragraph{Risks of Restaking.}
Restaking has attracted over \$20 billion in capital in 2024 alone~\cite{defillama-restaking}, serving as one of the fastest capital formation events within the history of cryptocurrency.
Much of this capital formation has arisen because long-term Ethereum holders view restaking as a means for enhancing their PoS yield with minimal excess risk.
However, restaking poses extra risks to users due to services' slashing rules.
%
A specific novel risk arises from service pooling.
Service pooling refers to a single operator using the same stake to operate multiple services.
For instance, a node operator might lock up 100 ETH of stake into $k$ services.
If each service provides yield $\gamma_1, \ldots, \gamma_k$, then the node operator can earn up to $\sum_i \gamma_i$ yield.
On the other hand, if the user is slashed on \emph{any} service, then their total staked quantity goes down for all services.
If the same node operator is slashed on service $i$ for 10 ETH, then the operator has 90 ETH staked on all $k$ services.

Pooling shares slashing risk across all services with common node operators, which implies that services themselves indirectly bear risk from other services.
The worst-case outcome is a cascading attack. 
This is where slashes in one service impacts other services that are pooled via common node operators.
If a malicious node operator is willing to be slashed in order to earn a profit from corrupting the network, groups of services can be attacked sequentially until the entire network's stake is slashed.
The seminal work~\cite{naveen-placeholder} demonstrated that this can happen given an adversary can attack arbitrarily large groups of services and if the network is not sufficiently overcollateralized.

\paragraph{Prior Work.}
The Eigenlayer restaking network~\cite[App. B]{eigen-wp} was the first to address restaking risks.
This paper focused on computing how much honest stake is needed to secure a service $s$ that has a maximum profit from adversarial behavior, $\pi_s$.
While this paper provides a polynomial time algorithm for detecting if a network can be exploited given $\pi_s$, it does not provide any formal guarantees on the losses of stake under attacks.

Subsequent work~\cite{naveen-placeholder} considered the problem of measuring cascades by representing restaking networks with bipartite graphs.
These graphs represent the relationship between services and node operators.
Properties of this graph and $\pi_s$ can lead to cascades, with~\cite{naveen-placeholder} constructing an infinite family of graphs that have a worst-case cascade (\ie~all of the stake is destroyed via correlated slashes).
On the other hand, the paper proves that if the network is overcollateralized in a particular sense (see~\S\ref{sec:model}), then the size of the largest cascade decays.
However, it should be noted that the overcollateralization is global, \ie~services require overcollateralization that depends on arbitrarily numbers of other services.

Finally, there have been a number of works on analyzing the effect of token incentives to impact PoS network security.
These papers analyzed concentration of wealth effects~\cite{fanti2019compounding}, competition between PoS and application yield~\cite{Chitra2021Competitive}, and the principal-agent problem with liquid staking (which are also popular within restaking)~\cite{chitra2020stake, tzinas2023principal}.
These works are related to this paper as they analyze the interaction between economic incentives and network security.

\subsection{Our results}
We expand the model of~\cite{naveen-placeholder} in two main ways:
\begin{enumerate}
    \item Inclusion of incentives paid by services to attract node operators
    \item Expand the types of adversarial attacks possible
\end{enumerate}

\paragraph{Realistic Adversaries.}
We show that for a realistic adversary, which we term a \emph{strictly submodular adversary}, one can choose rewards to ensure that cascades are bounded.
These adversaries realize decreasing marginal returns for attacking larger sets of services.
Bounded cascade sizes imply adversaries cannot execute a sequence of attacks that leads to the entire network being slashed.
Such a bound is important for analyzing how a restaking network implicitly affects the security of its host PoS network.
We show that for such an adversary, the length of a cascade degrades to the minimum length as the number of services grows to infinity.
Our results provide a more optimistic view on security against cascading failures as it is substantially weaker than a global overcollateralization condition (\ie~$\gamma$-security~\cite[Thm. 1]{naveen-placeholder}; see Appendix~\ref{app:overlap}).

We note that our model of strictly submodular adversaries represents a realistic model where the cost of attacking multiple services grows as more services are attacked concurrently.
For instance, if an attacker needs to aggregate stake across $k$ services and has to purchase at least $\sigma$ units of stake for each service, they will push the price up of the staking asset in order to execute the attack.
This, in particular, will lower the profitability of the attack as $k$ increases, potentially restricting the number of services that can be attacked simultaneously.

\paragraph{Incentives and Strategic Operators.}
Our model relies on more than strictly submodular adversaries.
We also require node operators to rebalance or adjust which services they are restaking with.
Node operators are modeled as strategic, adjusting their allocation to services based on the expected profit they receive via service incentive payments.
This is also realistic given that liquid restaking protocols (who make up over 50\% of restaked capital~\cite{defillama-restaking}) employ strategies to optimize their allocation to services~\cite{walter-avs, neuder-chitra-LRT, gauntlet-lrt}.

These incentive payments can be viewed as analogous to block rewards that are paid out as a subsidy to attract stakers in PoS networks.
Akin to work on PoS networks that shows that block rewards need to be sufficiently high to ensure that networks have sufficient stake to avoid attacks~\cite{Chitra2021Competitive}, we demonstrate that with sufficiently high rewards, one can ensure that node operators rebalance in a manner that reduces cascades.
We note that services and liquid restaking tokens on Eigenlayer have already paid out tens of millions of dollars of incentives  far~\cite{llama-risk-rewards, eigenlayer-rewards}.

\paragraph{Threat Models and Algorithms for Optimal Incentives.}
One can view the choice of a submodular adversary as a choice of threat model for feasible attacks.
The main model of strictly submodular adversary studied within this paper is the $\ell_p$-adversary.
This adversary faces weighted $p$-norm costs for attacking $k$ services out of a set of $S$ possible services.
When an adversary has this profit function, it implicitly means that an adversary cannot attack more than $O(S^{1/p})$ services simultaneously.
Note that this implies that for $p \rightarrow \infty$, we degrade to the threat model of $S$ independent PoS networks (\ie~we assume an adversary can only attack 1 network at a time).

Our main result in Theorem 1 shows that there exist sufficiently high rewards (incentives) $r_s(p)$ that can be paid to each service $s$ to ensure that the cascade length is bounded under the assumption of $\ell_p$ adversaries. This result implies that services can individually and locally choose a risk tolerance (parametrized by $p \in (1, \infty)$ and pay rewards to ensure they have no large cascades. Given that submodular functions are known to have minima that are easy to approximate~\cite{alaei2021maximizing, patton2023submodular}, a natural question is if there exist algorithms for computing the optimal rewards to distribute given a choice of $p$.
We show that this is indeed possible in~\S\ref{sec:algorithms} and provide an approximate guarantee dependent on attack profitability and a choice of $p$.


\section{Model}\label{sec:model}
Analogous to~\cite{naveen-placeholder}, we define a restaking graph as a bipartite graph with associated profit, stake, and threshold functions.
These functions will be used to define what it means for a restaking network to be secure to cascading risks.
Our model generalizes that of~\cite{naveen-placeholder} in that we consider a larger set of profit functions and we introduce a notion of rewards that services can pay to node operators and costs that node operators face for operating a service.
Our model is sufficiently general to handle both deterministic costs (\ie~cost of running hardware) and probabilistic costs (\ie~cost of being slashed). 

\paragraph{Restaking Graphs.}
A \emph{restaking graph} $G = (S, V, E, \sigma, \pi, \alpha, f)$ consists of
\begin{itemize}
    \item Bipartite graph with vertex set $S \sqcup V$ where $S$ is the set of services\footnote{In Eigenlayer terminology, a service would be called an `actively validated service' (AVS)} and $V$ is the set of node operators
    \item An edge $(v, s) \in E \subset V \times S$ if node operator $v$ is a node operator for service $s$
    \item $\sigma \in \reals^V_+$ is the amount of stake that node operator $v \in V$ has in the network
    \item $\pi \in \reals_+^S$ is the maximum profit from corruption that can be realized for each service $s \in S$
    \item $\alpha \in [0, 1]^S$ is the threshold percentage of stake that needs to collude to corrupt the service $s \in S$ (\eg~$\alpha = \nicefrac{1}{3}$ is the threshold for a BFT service)
    \item $f : \reals_+^S \times 2^{S} \rightarrow \reals_+$ is a profit function, where $f(\pi, A)$ is the maximum profit that can be realized by corrupting all services $s \in A$ simultaneously
\end{itemize}
See Figure~\ref{fig:restaking-graph} for a picture  of a restaking graph.
Our definition generalizes~\cite{naveen-placeholder} since they restrict attention to the linear profit function $f(\pi, A) = \sum_{s\in A} \pi_s$.
For a node operator $v$, we define its neighbor set (or boundary) as $\partial v = \{ s : (v,s) \in E \}$.
Similarly, for a service $s$, we define its neighborhood as $\partial s = \{ v : (v, s) \in E\}$.
For each service, we define the total stake as service $s$, $\sigma_{\partial s}$ as
\[
   \sigma_{\partial s} = \sum_{v : (v,s) \in E} \sigma_v
\]
We define the set $D_{\psi}(G)$ as the set of coalitions of node operators with stake less than $\psi \in (0, 1)$ fraction of the total amount staked:
\[
D_{\psi}(G) = \left\{ D \subseteq V \bigg\vert \sum_{v \in D} \sigma_v \leq \psi \sum_{v\in V} \sigma_v \right\}
\]
For any set $D \subset V$ or set $A \subset S$, we will slightly abuse notation and write
\begin{align*}
\sigma_D = \sum_{v \in D} \sigma_v && \pi_A = \sum_{s \in A}\pi_s
\end{align*}
Finally, for any set $A \subset S$ and a vector $v \in \reals^S$, we denote by $v(A) \in \reals^A$ the restriction of $v$ to the coordinate in $A$ (and similarly for $B \subset V)$. 

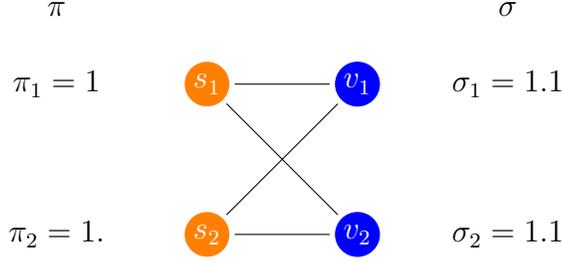
\begin{figure}
\begin {center}
\begin{tikzpicture}
  \node at (-2, 2) {\(\pi\)};
  \node at (4, 2) {\(\sigma\)};

  \coordinate (A1) at (0,1);
  \coordinate (A2) at (0,-1);
  \coordinate (B1) at (2,1);
  \coordinate (B2) at (2,-1);

  \node[fill=orange, circle, inner sep=1.5pt, outer sep=2pt, text=white, font=\bfseries] (A1) at (0,1) {$s_1$};
  \node[fill=orange, circle, inner sep=1.5pt, outer sep=2pt, text=white, font=\bfseries] (A2) at (0,-1) {$s_2$};
  \node[fill=blue, circle, inner sep=1.5pt, outer sep=2pt, text=white, font=\bfseries] (B1) at (2,1) {$v_1$};
  \node[fill=blue, circle, inner sep=1.5pt, outer sep=2pt, text=white, font=\bfseries] (B2) at (2,-1) {$v_2$};

  \draw (A1) -- (B1);
  \draw (A1) -- (B2);
  \draw (A2) -- (B1);
  \draw (A2) -- (B2);

  \node at (-2, 1) {\(\pi_1= 1\)};
  \node at (-2, -1) {\(\pi_2 = 1.\)};
  \node at (4, 1) {\(\sigma_1 = 1.1\)};
  \node at (4, -1) {\(\sigma_2 = 1.1\)};
\end{tikzpicture}
\caption{
Example of a restaking graph $G = (S, V, E, \alpha, \sigma, \pi)$ with $S = \{s_1, s_2\}, V = \{v_1, v_2\}$, $E = \{(s_1, v_1), (s_1, v_2), (s_2, v_1), (s_2, v_2)\}$.
We consider $f(\pi, A) = \sum_{s\in A} \pi_s$ as the profit function.
Note that each individual service cannot be attacked here as $\pi_i < \sigma_i$ for $i \in \{1, 2\}$.
However, the set $S$ might be vulnerable since the profitability condition~\eqref{eq:costly-profit}, $\pi_1 + \pi_2 > \sigma_i$ holds for the potential attack $(\{s_1, s_2\}, \{v_i\})\subset S \times V$.
This attack is only valid, however, if $\sigma_i > \alpha_{s_j} (\sigma_1 + \sigma_2)$, which implies that we need to have $\alpha_{s_j} < \frac{1}{2}$ for $j\in \{1, 2\}$ for this to be an attack.  
So if $s_1, s_2$ were BFT protocols with $\alpha_s = \frac{1}{3}$, this graph would be insecure.
However, if it they were longest-chain protocols with $\alpha_s = \frac{1}{2}$, it would be secure.
}
\label{fig:restaking-graph}
\end{center}
\end{figure}

\paragraph{Security and Overcollateralization.}
A restaking graph $G$ has an $f$-\emph{attack} at $(A,B) \subset S \times V$ if:\footnote{These conditions were originally identified in the Eigenlayer whitepaper \cite{eigen-wp} for the special case that $f(\pi, A) = \pi_A$.}
\begin{align}
    f(\pi, A) &\geq \sum_{v\in B}\sigma_v = \sigma_B \label{eq:costly-profit}\\
    \forall s \in A: \;\; \sum_{v \in B \cap \partial s} \sigma_v &\geq \alpha_s \sum_{v \in \partial s} \sigma_v = \alpha_s \sigma_{\partial s} \label{eq:feasibility}
\end{align}
When the context is clear, we will refer to an $f$-attack simply as an attack:
\begin{itemize}
    \item We call~\eqref{eq:costly-profit} the \emph{profitability} condition for an attack, as it requires that the net profit from corruption of a set of services $A$ exceed the total amount staked by the attacking operators $B.$
    \item We refer to~\eqref{eq:feasibility} as the \emph{feasibility} condition for an attack as it represents a coalition $B \subset V$ having sufficient stake to execute an attack.
\end{itemize}
A graph is said to be \emph{secure} if there does not  exist an $f$-attack $(A, B) \subset S \times V$ .
See Figure~\ref{fig:restaking-graph} for an example of a graph that is secure for $\alpha_s \geq \frac{1}{2}$ and insecure otherwise.

A restaking graph is said to be \emph{$\gamma$-secure} if $G$ is secure and for all attacking coalitions $(A, B) \subset S \times V$ (\eg~where~\eqref{eq:feasibility} is feasible):
\begin{equation}\label{eq:gamma-def}
(1+\gamma)f(\pi, A) \leq \sum_{v\in B}\sigma_v = \sigma_B
\end{equation}
This is an overcollateralization condition providing a multiplicative gap between the profit over attacking $A$ and the stake held by $B$.

As per~\cite{naveen-placeholder}, given an attack $(A, B)$, we define the graph $G \searrow B$ to be the subgraph $G = (S - A, V - B, E - (S \times B \cup A \times V), \sigma(V-B), \pi(S-A), \alpha(S-A), f)$. 
This is simply the restaking graph where the services in $A$ and the node operators in $B$ are removed.
Note that we use a slightly different definition relative to~\cite{naveen-placeholder} in that we remove services that have been attacked to simplify the dynamics of our model.
A disjoint sequence $(A_1, B_1), \ldots, (A_T, B_T)$ is a \emph{cascading sequence of attacks} if for each $t \in [T]$, $(A_t, B_t)$ is a valid attack on $G \searrow B_1 \cdots \searrow B_{t-1}$.
We let $C(G)$ denote the set of sequences of valid attacks on a restaking graph $G$ and as per~\cite{naveen-placeholder}, we define the cascade coefficient $R_{\psi}(G)$ as
\[
R_{\psi}(G) = \psi + \max_{D \in D_{\psi}(G)} \max_{(A_1, B_1),\ldots,(A_T, B_T) \in C(G \searrow D)} \frac{\sigma_{\bigcup_{t=1}^T B_t}}{\sigma_V}
\]
One can interpret $R_{\psi}(G)$ as the maximum loss of stake that can occur if stake of at most $\psi$ is slashed or removed.
In Figure~\ref{fig:cascading-attack}, we show a cascading attack of length $T=4$ with $R_{1/V}(G) = 1$ that is inspired by~\cite[Theorem 7]{naveen-placeholder}.

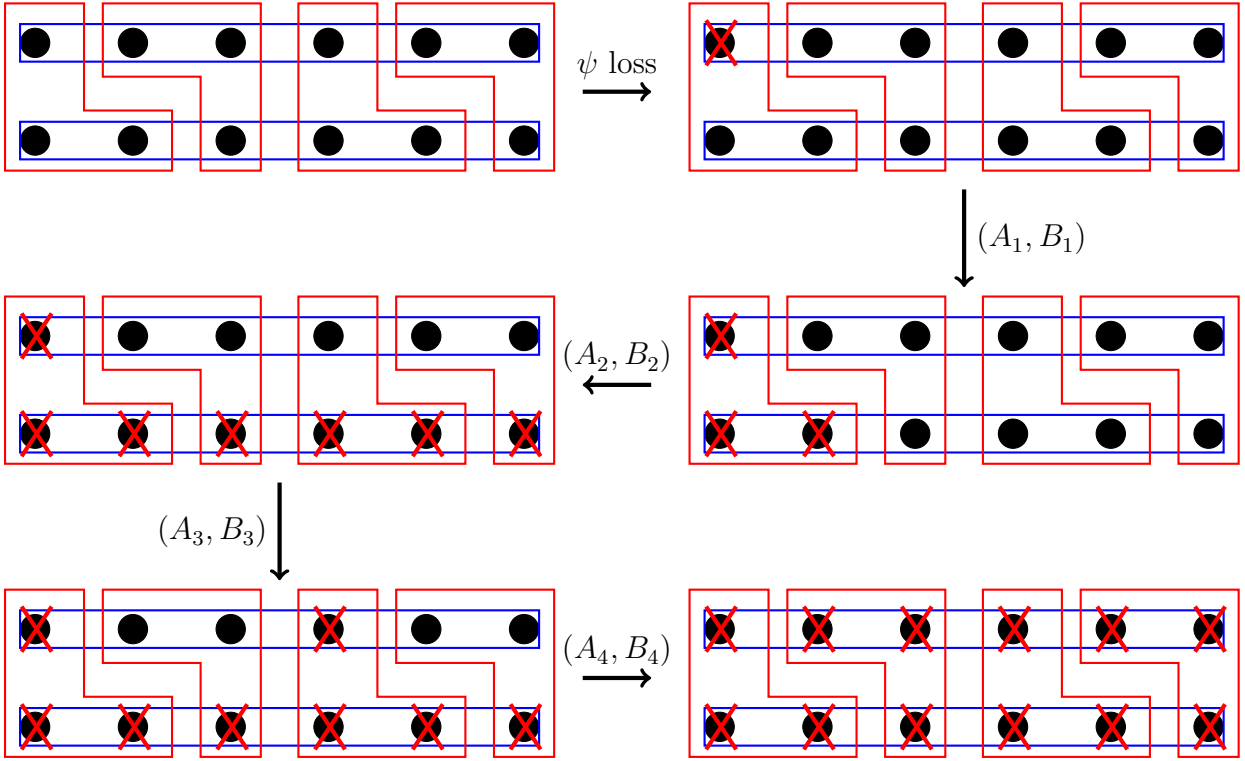
\begin{figure}
\begin{center}

\begin{tikzpicture}
    \def\dist{1.3} 
    \def\nodesize{0.2}
    \def\rectwidth{\dist*6}
    \def\rectheight{0.5}
    \def\Lwidth{\dist*1.5}
    \def\Lheight{\dist*1.5}
    \def\xoffset{\dist*7}
    \def\yoffset{\dist*3}

    \newcommand{\drawdiagram}[2]{
        \def\argone{#1}%
        \def\argtwo{#2}%

        \begin{scope}[shift={(#1, #2)}]
            \foreach \i in {0, 1, 2, 3, 4, 5} {
                \node[circle, draw, fill=black, minimum size=\nodesize cm] at (\i*\dist, 0) {};
            }

            \foreach \i in {0, 1, 2, 3, 4, 5} {
                \node[circle, draw, fill=black, minimum size=\nodesize cm] at (\i*\dist, -\dist) {};
            }

            \draw[blue, thick] (-\nodesize, \rectheight/2) rectangle (\dist*5+\nodesize, -\rectheight/2);
            \draw[blue, thick] (-\nodesize, -\dist+\rectheight/2) rectangle (\dist*5+\nodesize, -\dist-\rectheight/2);

            \foreach \i in {0, 3} {
                \draw[red, thick] (\i*\dist-2*\nodesize, \nodesize) -- (\i*\dist-2*\nodesize, 0.25*\dist+\nodesize) -- (\i*\dist+0.5*\dist, 0.25*\dist+\nodesize) -- (\i*\dist+0.5*\dist, -\dist+2*\nodesize) -- (\i*\dist+1.4*\dist, -\dist+2*\nodesize) -- (\i*\dist+1.4*\dist, -\dist-2*\nodesize) -- (\i*\dist-2*\nodesize, -\dist-2*\nodesize) -- cycle;
            }
            \foreach \i in {1, 4} {
                \draw[red, thick] (\i*\dist-2*\nodesize, \nodesize) -- (\i*\dist-2*\nodesize, 0.25*\dist+\nodesize) -- 
                (\i*\dist+\dist+2*\nodesize, 0.25*\dist+\nodesize) --
                (\i*\dist+\dist+2*\nodesize, -\dist-2*\nodesize) -- (\i*\dist+\dist-2*\nodesize, -\dist-2*\nodesize) -- (\i*\dist+\dist-2*\nodesize, -0.5*\dist+\nodesize) -- (\i*\dist-2*\nodesize, -0.5*\dist+\nodesize) -- cycle;
            }
        \end{scope}
    }

    \foreach \row in {0, 1, 2} {
        \foreach \col in {0, 1} {
            \drawdiagram{\col*\xoffset}{-\row*\yoffset};
        }
    }

    \draw[xshift=\xoffset, yshift=0, ultra thick, red] (\xoffset-2.5*\nodesize, 0.3) -- ++(0.4, -0.6);
    \draw[xshift=\xoffset, yshift=0, ultra thick, red] (\xoffset-2.5*\nodesize, -0.3) -- ++(0.4, 0.6);

    \draw[xshift=\xoffset, yshift=0, ultra thick, red] (\xoffset-2.5*\nodesize, 0.3-\yoffset) -- ++(0.4, -0.6);
    \draw[xshift=\xoffset, yshift=0, ultra thick, red] (\xoffset-2.5*\nodesize, -0.3-\yoffset) -- ++(0.4, 0.6);
    \draw[xshift=\xoffset, yshift=0, ultra thick, red] (\xoffset-2.5*\nodesize, 0.3-\dist-\yoffset) -- ++(0.4, -0.6);
    \draw[xshift=\xoffset, yshift=0, ultra thick, red] (\xoffset-2.5*\nodesize, -0.3-\dist-\yoffset) -- ++(0.4, 0.6);
    \draw[xshift=\xoffset, yshift=0, ultra thick, red] (\xoffset-2.5*\nodesize+\dist, 0.3-\dist-\yoffset) -- ++(0.4, -0.6);
    \draw[xshift=\xoffset, yshift=0, ultra thick, red] (\xoffset-2.5*\nodesize+\dist, -0.3-\dist-\yoffset) -- ++(0.4, 0.6);

    \draw[xshift=\xoffset, yshift=0, ultra thick, red] (-2.5*\nodesize, 0.3-\yoffset) -- ++(0.4, -0.6);
    \draw[xshift=\xoffset, yshift=0, ultra thick, red] (-2.5*\nodesize, -0.3-\yoffset) -- ++(0.4, 0.6);
    
    \foreach \i in {0, 1, 2, 3, 4, 5} { 
        \draw[xshift=\xoffset, yshift=0, ultra thick, red] (-2.5*\nodesize + \i*\dist, 0.3-\dist-\yoffset) -- ++(0.4, -0.6);
        \draw[xshift=\xoffset, yshift=0, ultra thick, red] (-2.5*\nodesize + \i*\dist, -0.3-\dist-\yoffset) -- ++(0.4, 0.6);
    }

    \foreach \i in {0, 3} {
        \draw[xshift=\xoffset, yshift=0, ultra thick, red] (-2.5*\nodesize+\i*\dist, 0.3-2*\yoffset) -- ++(0.4, -0.6);
        \draw[xshift=\xoffset, yshift=0, ultra thick, red] (-2.5*\nodesize+\i*\dist, -0.3-2*\yoffset) -- ++(0.4, 0.6);
    }
    
    \foreach \i in {0, 1, 2, 3, 4, 5} { 
        \draw[xshift=\xoffset, yshift=0, ultra thick, red] (-2.5*\nodesize + \i*\dist, 0.3-\dist-2*\yoffset) -- ++(0.4, -0.6);
        \draw[xshift=\xoffset, yshift=0, ultra thick, red] (-2.5*\nodesize + \i*\dist, -0.3-\dist-2*\yoffset) -- ++(0.4, 0.6);
    }

    \foreach \i in {0, 1, 2, 3, 4, 5} {
        \draw[xshift=\xoffset, yshift=0, ultra thick, red] (\xoffset-2.5*\nodesize+\i*\dist, 0.3-2*\yoffset) -- ++(0.4, -0.6);
        \draw[xshift=\xoffset, yshift=0, ultra thick, red] (\xoffset-2.5*\nodesize+\i*\dist, -0.3-2*\yoffset) -- ++(0.4, 0.6);
    }
    
    \foreach \i in {0, 1, 2, 3, 4, 5} { 
        \draw[xshift=\xoffset, yshift=0, ultra thick, red] (\xoffset-2.5*\nodesize + \i*\dist, 0.3-\dist-2*\yoffset) -- ++(0.4, -0.6);
        \draw[xshift=\xoffset, yshift=0, ultra thick, red] (\xoffset-2.5*\nodesize + \i*\dist, -0.3-\dist-2*\yoffset) -- ++(0.4, 0.6);
    }

    \draw[->, ultra thick] (0.8*\xoffset, -\dist/2) -- (0.9*\xoffset, -\dist/2) node[midway, above] {$\psi$ loss};
    \draw[->, ultra thick] (\xoffset+2.5*\dist, -1.5*\dist) -- (\xoffset + 2.5*\dist, -2.5*\dist) node[midway, right] {$(A_1, B_1)$};
    \draw[<-, ultra thick] (0.8*\xoffset, -\dist/2-\yoffset) -- (0.9*\xoffset, -\dist/2-\yoffset) node[midway, above] {$(A_2, B_2)$};
    \draw[->, ultra thick] (2.5*\dist, -1.5*\dist-\yoffset) -- (2.5*\dist, -2.5*\dist-\yoffset) node[midway, left] {$(A_3, B_3)$};
    \draw[->, ultra thick] (0.8*\xoffset, -\dist/2-2*\yoffset) -- (0.9*\xoffset, -\dist/2-2*\yoffset) node[midway, above] {$(A_4, B_4)$};
\end{tikzpicture}

\caption{
An example of a cascading failure in a restaking network based on~\cite[Thm. 7, Figure 1]{naveen-placeholder}
In this sequence of figures, black dots represent validators $v_i, i \in [12]$ and the red and blue boxes containing $v_i$ represent services $s_i, i \in [6]$.
For this system, we have $\sigma_{v_i} = 1$ and $\alpha_{s_i} = 1$ for all $s_i, v_i$.
We have $\pi_s = 2$ for the services represented by red boxes and $\pi_s = 4$ for the services represented by the blue boxes.
One can view this as the hypergraph representation of a bipartite graph implicit in a restaking graph.
When we go from the upper left diagram to the upper right diagram, we first have a loss of $\psi = \frac{1}{12}$, which represents the loss of a single node's stake.
Losing this node's stake makes the left red box vulnerable as $\pi_s = 2\sigma_v$ when $s$ is a red box, which is represented when one goes from the upper right box to the middle right box in attack $(A_1, B_1)$.
This leads to the bottom row becoming vulnerable as $\pi_s = 4\sigma_v$ when $s$ is a blue box (which is the attack $(A_2, B_2)$.
This attack now leads the middle red service vulnerable and it is attacked in $(A_3, B_3)$.
Finally, the remaining nodes in the blue service in the top row are vulnerable and attacked in $(A_4, B_4)$, leading to $R_{1/12}(G) = 1$
}

\label{fig:cascading-attack}
\end{center}
\end{figure}

In~\cite{naveen-placeholder}, the authors show that when $f(\pi, A) = \sum_{s \in A} \pi_s$, one can have $R_{\psi}(G) = 1$, which represents the entire network being slashed in a cascading failure.
However, the authors also demonstrate that if $G$ is $\gamma$--secure, then we have
\begin{equation}\label{eq:naveen-cascade}
R_{\psi}(G) \leq \left(1 + \frac{1}{\gamma}\right)\psi
\end{equation}
This demonstrates that networks are safer as they increase their overcollateralization.
Note, however, that $\gamma$ is a global variable for overcollateralization (\ie~every service needs to be overcollateralized by at least this amount) which can make service profitability difficult.
Finally, we note~\cite[Theorem 1, Corollary X]{naveen-placeholder},~\cite[App. B]{eigen-wp} that a computationally efficient sufficient condition for Eigenlayer nodes to be $\gamma$-secure is: $\forall v \in V, \;(1+\gamma) \sum_{s\in \partial v} \frac{\pi_s}{\alpha_s\sigma_{\partial s}} \leq 1$.
Note that condition need not hold for all $\gamma$-secure restaking graphs.

\paragraph{Overlap.}
One of the main advantages of restaking systems is the ability to reuse capital to secure multiple services.
This allows node operators to earn yield from many sources and increase their overall profitability.
However, on the flip side, this also increases the profit from attack.
As a simple example, consider the restaking graph of Figure~\ref{fig:restaking-graph}.
In this graph, we have two services $a$ and $b$ with $\pi_a = \pi_b = 1$.
Moreover, there is a single node operator $\sigma_1$ who is validating both services $a$ and $b$.
As $\sigma_1 = 1.1$, then on their own, neither service $a$ nor $b$ are attackable as $\pi_i - \sigma_1 < 0$ for $i \in \{a, b\}$.
However, if $f(\pi, A) = \sum_{s\in A} \pi_s$, then it is profitable to attack both $a$ and $b$ simultaneously since $\pi_a + \pi_b - \sigma_1 > 0$.
On the other hand, if there were two node operators with $\sigma_1 = \sigma_2 = 1.1$ with $\sigma_1$ operating $a$ and $\sigma_2$ operating $b$, then there is no viable attack.
This example demonstrates that when two services share stake operated by the same operator, they increase the profitability of attacking both services simultaneously.

One can view the risk of being attacked as related to the \emph{overlap} of stake between services $s$ and $t$.
We define the overlap between $s$ and $t$, $\theta_{s,t}$ as $\theta_{s,t} = \sigma_{\partial s \cap \partial t}$ where $\partial s \cap \partial t \subset V$ is the set of operators validating both services.
In particular, as two services have higher overlap, the profit from attacking them together (as in our example) goes up.
For each service $s \in S$, we define the minimum and maximum overlap, $\utheta_s, \otheta_s$, as 
$\utheta_s = \min_{t\in S/\{s\}} \theta_{s, t}, \otheta_s = \max_{t\in S/\{s\}} \theta_{s, t}$.

In Appendix~\ref{app:overlap}, we demonstrate an example of a graph whose maximum cascade length is controlled completely by the overlap.
For this graph, if one required $\gamma$-security to hold, we demonstrate that one would need $\Omega(\sum_{s\in S} \pi_s)$ more stake than necessary to have small cascade length.
In a scenario where one service has a profit $\pi_1 = 1$ and another service has $\pi_2 = 1,000$, this implies that $\pi_1$ has to aggregate stake on the order of $1,000$ times more stake to be $\gamma$-secure than if they were an isolated network.
This example demonstrates that $\gamma$-security is often a highly inefficient means of overcollateralizing a restaking system to avoid cascades.

One of the main insights of this paper is that if one can control the overlap suitably and rational node operators can adjust their stake in response to an attack, then one can bound cascade length without needing a condition as strong as being $\gamma$-secure.
In the sequel, we demonstrate that control over the minimum and maximum overlap between services allows one to bound the profitability of an attack.
For instance, when the profitability is sublinear in $|A|$, then constraints on the overlap between services upper and lower bound $f(\pi, A) - \sigma_B$.
Our results in \S\ref{sec:main-results} show that these bounds allow for incentives to be used to control over the overlap and $R_{\psi}(G)$.

\paragraph{Strictly Submodular Adversaries.}
We say that a restaking graph $G = (S, V, E, \sigma, \pi, \alpha, f)$ is strictly submodular if $f(\pi, A)$ is a strictly submodular function of $A$,\ie~$,
f(\pi, A\cup A') + f(\pi, A\cap A') < f(A) + f(A')$
for all $A, A' \subset S$ with $A \not\subset A', A' \not\subset A$.
A submodular adversary can be viewed as an adversary who faces costs for simultaneously attacking many services.
Note that the linear payoff function $f(\pi, A) = \sum_{s \in A} \pi_s$ of~\cite{naveen-placeholder} does not yield a strictly submodular restaking graph.
Strictly submodular restaking graphs can be thought of as those where attackers face a cost that increases with the number of services $|A|$ they are trying to concurrently attack.
For instance, if an attacker faces an acquisition cost for acquiring stake in each service they attack, they could have a profit that is sublinear in $|A|$.
We say that an attack is a \emph{costly attack} if $f$ is strictly submodular.

For example, suppose an adversary faces a multiplicative cost $c(\pi, A)$ for attacking services $A \subset S$, \ie~$f(\pi, A) = c(\pi, A) \left(\sum_{s \in A} \pi_s\right)$.
If for all $s$, $c(\pi, A) = \Theta\left(\frac{1}{|A|^c}\right)$ for $c$ (\ie~non-constant cost, increasing in $A$), then
$f(\pi, A) = O(|A|^{1-c})$, which is strictly submodular~\cite[\S2]{patton2023submodular}.
Such a multiplicative cost might arise if one has an additively supermodular cost (which can be thought of as a strictly convex cost).

This is also equivalent to having diminishing returns for attacking multiple services.
This occurs when the costs for acquiring stake to attack different services are correlated.
For instance, acquiring stake to attack service $A$ can cause the price of acquiring stake to attack service $B$ to increase.
Such non-trivial costs growing with the maximum profit exist in other blockchain attacks: oracle manipulation~\cite{mackinga2022twap, aspembitova2022oracles}, intents~\cite{chitra2024analysis}, transaction fee manipulation~\cite{yaish2023speculative, chung2024collusion}, liquid staking manipulation~\cite{chitra2020stake,tzinas2023principal} and time-bandit attacks~\cite{yaish2022blockchain}. 

We also note that if $f(\pi, A)$ is a weighted $\ell_p$-norm of $A$, then it is strictly submodular if $p > 1$~\cite[Obs. 2.1]{patton2023submodular}.
The majority of this paper will focus on studying the $p$-norm profit, $f_p(\pi, A)$:
\begin{equation}\label{eq:lp-profit}
f_p(\pi, A) = \left(\sum_{s\in A} \pi_s^p\right)^{1/p}
\end{equation}

We note that for strictly submodular adversaries, Corollary 1 of~\cite{naveen-placeholder} does not hold.
When $f(\pi, A) = \sum_{s \in A} \pi_s$, this corollary says that if $(A_1, B_1), \ldots, (A_T, B_T)$ is a valid attack sequence then $(\bigcup_t A_t, \bigcup_t B_t)$ is a valid attack sequence.
We demonstrate that this is not true for $p$-attacks in Appendix~\ref{app:unions}.

\paragraph{Incentives.}
We say a restaking graph is \emph{incentivized} if there are two additional functions:
\begin{itemize} 
    \item $r \in \reals^S_+$: $r_s$ is the reward paid out by service $s$ to node operators $\partial s$ 
    \item $c \in \reals^S_+$: $c_s$ is the cost to an operator of operating service $s$
\end{itemize}
Costs should be thought of as including both the cost of operating the service and the losses faced from not conforming to the protocol (\ie~slashes).
Given a reward $r_s$, each node operator $v$ with an edge $(v, s) \in E$ receives a reward $\rho_{sv} = \beta_{sv} r_s$.
If $(v, s) \not\in E$, we define $\rho_{sv} = 0$.
A service $s$ is \emph{pro-rata} if $\beta_{sv} = \frac{\sigma_v}{\sigma_{\partial s}}$, \ie~$
\rho_{sv} = r_s \frac{\sigma_v}{\sum_{v : (v,s) \in E} \sigma_v} = r_s \frac{\sigma_v}{\sigma_{\partial s}}$.
The majority of reward systems within decentralized networks and cryptocurrencies follow a pro-rata reward distribution~\cite{johnson2023concave, chen2019axiomatic} and this is the main form of incentive studied in this paper.
Finally, we say an incentivized restaking graph is \emph{dynamic} if $r(t) \in \reals^S_+, \sigma(t) \in \reals^V_+, E_t \subset 2^{S \times V}$ are updated sequentially for rounds $t \in \naturals$.

A node operator is \emph{profitable without impact} at a service $s$ if $\rho_{sv} - c_s = r_s \frac{\sigma_v}{\sigma_{\partial s}} - c_s \geq 0$ which is equivalent to:
\begin{equation}\label{eq:profitable-no-impact}
\sigma_v \geq \frac{c_s \sigma_{\partial s}}{r_s}
\end{equation}
However, a node operator also needs to measure the impact of adding their own stake when computing profitability.
Recall that upon depositing $\sigma_v$ into service $s$, $\sigma_{\partial s}$ will increase by $\sigma_v$.
A node operator is said to be \emph{profitable with impact} if $\sigma_v \geq \frac{c_s(\sigma_{\partial s} + \sigma_v)}{r_s}$ which can be rewritten as
\begin{equation}\label{eq:profitable-impact}
\sigma_v \geq \frac{c_s \sigma_{\partial s}}{r_s - c_s}
\end{equation}

\paragraph{Rebalancing.}
Consider a node operator $v \in V$ who is validating a set of services $\partial v(0) \subset S$ at time $t = 0$.
We say that $v$ is \emph{strategic} if given a dynamic, incentivized staking graph $G_t$ and an attack sequence $(A_1, B_1), \ldots, (A_T, B_T)$, they have $\partial v(t) = \mathsf{Update}(\partial v(t-1), \sigma_{\partial s}, r_s, c_s)$ after each attack $(A_t, B_t)$.
A strategic node operator can be thought of one that updates the services the validate in response to change in yield after $t$ attacks, $\rho_{sv}(t)$.
If that the function $\mathsf{Update} : 2^S \times \reals_+ \times \reals_+ \times \reals_+ \rightarrow 2^S$ is non-constant, we say that a node operator is \emph{strategically rebalancing}.

If we view $\partial v(t)$ as a portfolio of financial assets, a strategically rebalancing node operator is one who aims to maximize their expected value as the restaking graph changes.
For instance, if a slashing event reduces competition at a service $s$ so that $\rho_{sv}(t) > \rho_{sv}(t-1)$, then a strategically rebalancing node operator can add $s$ to $\partial v$ to earn more incentives.
When rebalancing is present, we consider a three-step iterated game between services, node operators, and attackers given an attack sequence:
\begin{enumerate}
\item Services choose rewards $r_s(t)$ to offer node operators
\item Attackers execute a single attack $(A_i, B_i) $
\item Strategically rebalancing node operators update the services they are validating, $\partial v(t)$
\end{enumerate}
We note that one can view this as an online Stackelberg security game with the services acting as leader and the attacker and node operators as followers~\cite{sinha2018stackelberg, balcan2015commitment}.

We construct an explicit example of rebalancing arresting a cascading attack in Figure~\ref{fig:rebalance}.
We take the same example as in Figure~\ref{fig:cascading-attack} and show that if a node operator can strategically rebalance and rewards are sufficiently high, they can halt a cascading attack.
In this example, we are able to reduce an attack $(A_1, B_1), \ldots, (A_T, B_T)$ to an attack $(A_1, B_1)$.
The main results of~\S\ref{sec:main-results} construct sufficient conditions for how high rewards need to be in order to halt an attack and force it to have length 1.

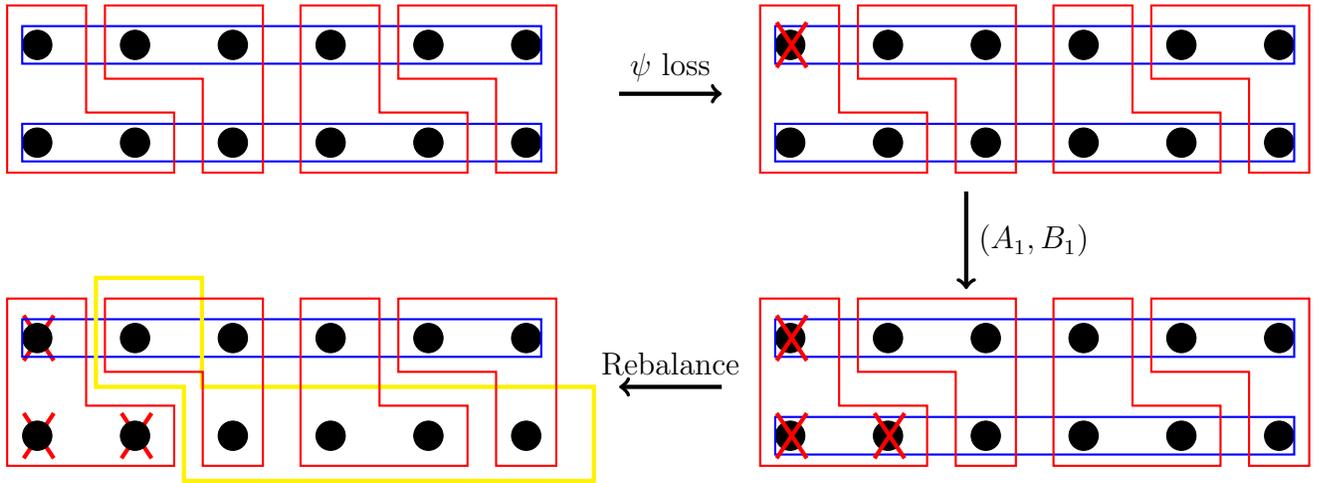
\begin{figure}
\begin{center}

\begin{tikzpicture}
    \def\dist{1.3}
    \def\nodesize{0.2}
    \def\rectwidth{\dist*6}
    \def\rectheight{0.5}
    \def\Lwidth{\dist*1.5}
    \def\Lheight{\dist*1.5}
    \def\xoffset{\dist*7}
    \def\yoffset{\dist*3}

    \newcommand{\drawdiagram}[2]{
        \begin{scope}[shift={(#1, #2)}]
            \foreach \i in {0, 1, 2, 3, 4, 5} {
                \node[circle, draw, fill=black, minimum size=\nodesize cm] at (\i*\dist, 0) {};
            }

            \foreach \i in {0, 1, 2, 3, 4, 5} {
                \node[circle, draw, fill=black, minimum size=\nodesize cm] at (\i*\dist, -\dist) {};
            }

            \draw[blue, thick] (-\nodesize, \rectheight/2) rectangle (\dist*5+\nodesize, -\rectheight/2);
            \draw[blue, thick] (-\nodesize, -\dist+\rectheight/2) rectangle (\dist*5+\nodesize, -\dist-\rectheight/2);

            \foreach \i in {0, 3} {
                \draw[red, thick] (\i*\dist-2*\nodesize, \nodesize) -- (\i*\dist-2*\nodesize, 0.25*\dist+\nodesize) -- (\i*\dist+0.5*\dist, 0.25*\dist+\nodesize) -- (\i*\dist+0.5*\dist, -\dist+2*\nodesize) -- (\i*\dist+1.4*\dist, -\dist+2*\nodesize) -- (\i*\dist+1.4*\dist, -\dist-2*\nodesize) -- (\i*\dist-2*\nodesize, -\dist-2*\nodesize) -- cycle;
            }
            \foreach \i in {1, 4} {
                \draw[red, thick] (\i*\dist-2*\nodesize, \nodesize) -- (\i*\dist-2*\nodesize, 0.25*\dist+\nodesize) -- 
                (\i*\dist+\dist+2*\nodesize, 0.25*\dist+\nodesize) --
                (\i*\dist+\dist+2*\nodesize, -\dist-2*\nodesize) -- (\i*\dist+\dist-2*\nodesize, -\dist-2*\nodesize) -- (\i*\dist+\dist-2*\nodesize, -0.5*\dist+\nodesize) -- (\i*\dist-2*\nodesize, -0.5*\dist+\nodesize) -- cycle;
            }
        \end{scope}
    }

    \drawdiagram{0}{0};
    \drawdiagram{1.1*\xoffset}{0};
    \drawdiagram{1.1*\xoffset}{-1*\yoffset};

    \draw[xshift=\xoffset, yshift=0, ultra thick, red] (1.1*\xoffset-2.5*\nodesize, 0.3) -- ++(0.4, -0.6);
    \draw[xshift=\xoffset, yshift=0, ultra thick, red] (1.1*\xoffset-2.5*\nodesize, -0.3) -- ++(0.4, 0.6);

    \draw[xshift=\xoffset, yshift=0, ultra thick, red] (1.1*\xoffset-2.5*\nodesize, 0.3-\yoffset) -- ++(0.4, -0.6);
    \draw[xshift=\xoffset, yshift=0, ultra thick, red] (1.1*\xoffset-2.5*\nodesize, -0.3-\yoffset) -- ++(0.4, 0.6);
      
    \foreach \i in {0, 1} { 
        \draw[xshift=\xoffset, yshift=0, ultra thick, red] (1.1*\xoffset-2.5*\nodesize + \i*\dist, 0.3-\dist-\yoffset) -- ++(0.4, -0.6);
        \draw[xshift=\xoffset, yshift=0, ultra thick, red] (1.1*\xoffset-2.5*\nodesize + \i*\dist, -0.3-\dist-\yoffset) -- ++(0.4, 0.6);
    }

    \draw[xshift=\xoffset, yshift=0, ultra thick, red] (-2.5*\nodesize, 0.3-\yoffset) -- ++(0.4, -0.6);
    \draw[xshift=\xoffset, yshift=0, ultra thick, red] (-2.5*\nodesize, -0.3-\yoffset) -- ++(0.4, 0.6);
    
    \foreach \i in {0, 1} { 
        \draw[xshift=\xoffset, yshift=0, ultra thick, red] (-2.5*\nodesize + \i*\dist, 0.3-\dist-\yoffset) -- ++(0.4, -0.6);
        \draw[xshift=\xoffset, yshift=0, ultra thick, red] (-2.5*\nodesize + \i*\dist, -0.3-\dist-\yoffset) -- ++(0.4, 0.6);
    }

    \draw[->, ultra thick] (0.85*\xoffset, -\dist/2) -- (\xoffset, -\dist/2) node[midway, above] {$\psi$ loss};
    \draw[->, ultra thick] (\xoffset+2.5*\dist, -1.5*\dist) -- (\xoffset + 2.5*\dist, -2.5*\dist) node[midway, right] {$(A_1, B_1)$};
    \draw[<-, ultra thick] (0.85*\xoffset, -\dist/2-\yoffset) -- (\xoffset, -\dist/2-\yoffset) node[midway, above] {Rebalance};

    \begin{scope}[shift={(0, -\yoffset)}]
            \foreach \i in {0, 1, 2, 3, 4, 5} {
                \node[circle, draw, fill=black, minimum size=\nodesize cm] at (\i*\dist, 0) {};
            }

            \foreach \i in {0, 1, 2, 3, 4, 5} {
                \node[circle, draw, fill=black, minimum size=\nodesize cm] at (\i*\dist, -\dist) {};
            }

            \draw[blue, thick] (-\nodesize, \rectheight/2) rectangle (\dist*5+\nodesize, -\rectheight/2);
            \draw[yellow, ultra thick] (6*\dist-2*\nodesize,-\dist-3*\nodesize) -- 
            (1.5*\dist,-\dist-3*\nodesize) --
            (1.5*\dist,-\dist/2) --
            (0.6*\dist,-\dist/2) --
            (0.6*\dist,4*\nodesize) --
            (1.3*\dist+2.5*\nodesize, 4*\nodesize) --
            (1.3*\dist+2.5*\nodesize, -\dist/2) --
            (6*\dist-2*\nodesize, -\dist/2) -- cycle;

            \foreach \i in {0, 3} {
                \draw[red, thick] (\i*\dist-2*\nodesize, \nodesize) -- (\i*\dist-2*\nodesize, 0.25*\dist+\nodesize) -- (\i*\dist+0.5*\dist, 0.25*\dist+\nodesize) -- (\i*\dist+0.5*\dist, -\dist+2*\nodesize) -- (\i*\dist+1.4*\dist, -\dist+2*\nodesize) -- (\i*\dist+1.4*\dist, -\dist-2*\nodesize) -- (\i*\dist-2*\nodesize, -\dist-2*\nodesize) -- cycle;
            }
            \foreach \i in {1, 4} {
                \draw[red, thick] (\i*\dist-2*\nodesize, \nodesize) -- (\i*\dist-2*\nodesize, 0.25*\dist+\nodesize) -- 
                (\i*\dist+\dist+2*\nodesize, 0.25*\dist+\nodesize) --
                (\i*\dist+\dist+2*\nodesize, -\dist-2*\nodesize) -- (\i*\dist+\dist-2*\nodesize, -\dist-2*\nodesize) -- (\i*\dist+\dist-2*\nodesize, -0.5*\dist+\nodesize) -- (\i*\dist-2*\nodesize, -0.5*\dist+\nodesize) -- cycle;
            }
        \end{scope}
    
\end{tikzpicture}

\caption{
A cascading attack halted by a rebalance.
This is the same example as in Figure~\ref{fig:cascading-attack}, except that rewards $r_s$ are chosen sufficiently high such that after the validator in the upper left corner is slashed, the adjacent validator in the upper blue service (\ie~second from the left), joins the lower blue service.
This leads to the cascading failure not being viable in that $(A_2, B_2)$ is no longer a valid attack after rebalancing.
One can view this as an equilibrium condition: if we start in equilibrium, then every $(v, s) \in E$ must be profitable with impact.
This implies that if they are slashed, then another validator with at most the same stake can profitably join the service after they're slashed (which is what the validator second from the left in the top row is doing).
}

\label{fig:rebalance}
\end{center}
\end{figure}


\section{Costly Attacks and Strategic Operators}\label{sec:main-results}
Given the model of the previous section, we are now ready to describe our main results.
We consider a weaker adversary and a stronger operator than~\cite{naveen-placeholder} in that we consider a dynamic, incentivized restaking graph $G_t = (S, V, E_t, \alpha, \sigma_t, \pi, f, r(t), c)$ where $f$ is strictly submodular.
For the remainder of this section, we will consider 
the $G^p_t = (S, V, E_t, \alpha, \sigma_t, \pi, f_p, r(t), c)$ where $f_p(\pi, A)$ is defined in~\eqref{eq:lp-profit}.

We will term a costly attack for $G_t^p$ a \emph{$p$-attack}. 
We prove similar results to those in this section for strictly submodular functions, where the analog of the parameter $p$ is curvature of a submodular function~\cite{sviridenko2017optimal}, in the full version of this paper.
Our main result can be stated as:
\begin{theorem}\label{thm:main}
Consider $G^p_t$ with strategically rebalancing node operators.
There exist rewards $r(t) \in \reals^S_+$ such that for a constant $C > 0$ we have
\[
R_{\psi}(G) \leq \psi + \frac{C}{S^{1-\frac{1}{p}}}
\]
The rewards $r(t)$ are local in that the optimal choice of $r_s(t)$ is only a function of $\partial s(t)$ and any service $s' \in S$ such that $\partial s'(t) \cap \partial s(t) \neq \emptyset$
\end{theorem}

\paragraph{Outline of Proof of Theorem~\ref{thm:main}.}
We prove this theorem in a sequence of steps:
\begin{enumerate}
\item We first show that for any costly attack $(A, B)$, one has
\[
|B| \leq K S^{1/p}
\]
for $K = \frac{\max_s \pi_s}{\min_v \sigma_v}$.
Note that $\min_v \sigma_v$ is a control parameter that a restaking protocol designer can choose (\ie~analogous to the 32 ETH that are needed to stake in Ethereum mainnet).

\item We next show that if $\forall s,t \in S,\;\theta_{s,t} = \Theta(S^{1/p})$ , then $|A| = o(S)$
\item For any costly attack sequence $(A_1, B_1), \ldots, (A_T, B_T) \in C(G^p_0)$, we show that if $|A_t| = o(S)$, $r_s = \Omega(S^{1/p})$, then rebalancing node operators only allow $(A_1, B_1)$ to be a valid attack.
That is, after the first attack $(A_1, B_1)$ is executed, rebalancing node operators place sufficient stake on services in a manner that makes $(A_2, B_2)$ an invalid attack on $G^p_1$.

\item Next, we show that given high enough rewards, rebalanced stake cannot cause further attacks.
That is, rebalancing cannot cause an attack that was infeasible prior to rebalancing to be feasible post rebalancing. 

\item Given these sufficiently high rewards, we also show that in the presence of strategically rebalancing node operators and if $V > S$, we have $\sigma_V \geq \left(\min_v \sigma_v\right) S$ which implies that cascades are bounded as:
\begin{equation}\label{eq:r-psi}
R_{\psi}(G) \leq \psi + \max_{D \in D_{\psi}(G)} \max_{(A_1, B_1) \in C(G-D)} \frac{\sigma_{B_1}}{\sigma_V} \leq \psi + \frac{K S^{1/p}}{\left(\min_v \sigma_v\right) S} = \psi + \frac{C}{S^{1-\frac{1}{p}}}
\end{equation}
\item All of the above results rely on ensuring that $\utheta_s, \otheta_s = \Theta(S^{1/p})$.
We demonstrate that there exists a simple update to dynamically choose rewards $r_s(t) = f(\utheta_s(t), \otheta_s(t), \sigma_{\partial s}(t), r_s(t-1))$
in response to rebalances to ensure that $\utheta_s \otheta_s = \Theta(S^{1/p})$.
This update rule gives lower rewards to node operators who increase overlap and higher rewards to those who decrease overlap.
We note that such incentive mechanisms have been used as so-called incentive `boosts' within various restaking point systems~\cite{llama-risk-rewards, eigenlayer-rewards}.
\end{enumerate}
Combined, these results ensure that one can ensure small cascade coefficient by dynamically updating rewards in the presence of strategically rebalancing node operators.

\subsection{Proofs of Steps 1 and 2}
The claim that $|B| = O(S^{1/p})$ for a $p$-attack is straightforward:
\begin{claim}[Step 1]\label{claim:localize-node operators}
Suppose that $\frac{\max_s \pi_s}{\min_v \sigma_v} \leq K$. If $(A, B)$ is $p$-attack, then $|B| \leq KS^{1/p}$  
\end{claim}
\begin{proof}
    For a $p$-attack $(A,B)$,~\eqref{eq:costly-profit} and~\eqref{eq:lp-profit} imply that $|A|^{1/p} \left(\max_{s\in A} \pi_s\right) \geq f_p(\pi, A) \geq \sigma_B \geq \left(\min_{v\in B} \sigma_v\right) |B|$.
    Therefore $|B| \leq \frac{\max_{s\in A}\pi_s}{\min_{v \in B}\sigma_v} |A|^{1/p} \leq K S^{1/p} $
\end{proof}
\noindent The second claim, \ie~$|A| = o(S)$, relies on upper and lower bounding the overlap between services.
\begin{claim}[Step 2]
    Suppose that there exists $\delta > 0$ such that for all $s, t \in S$, $|\partial s \cap \partial t| \geq (1+\delta)KS^{1/p}$ and $\otheta_s \leq \left(\max_v \sigma_v\right) K S^{1/p}$ for all $s \in S$.
    If $\frac{\delta \left(\min_v \sigma_v\right)}{\left(\max_v \sigma_v\right)} \geq 1-\frac{1}{e(S-1)^2}$, then for any  $p$-attack $(A, B)$, we have
    \begin{equation}\label{eq:A-bound}
    |A| \leq \left(\frac{2K(S-1)}{S-2}\right)^{\frac{p}{p-1}}
    \end{equation}
    where $K = \frac{\max_s \pi_s}{\min_v \sigma_v}$.
    Therefore if $K = o(S^{\frac{p-1}{p}})$, then $|A| = o(S)$
\end{claim}
This bound demonstrates the number of attacked services is sublinear in $S$ provided that the overlap is sufficiently small and that the adversary is submodular.
Furthermore, note that the assumption $|\partial s \cap \partial t| \geq (1+\delta)KS^{1/p}$ implies that $\utheta_s = \Omega(S^{1/p})$ since $\utheta_s \geq \left(\min_v \sigma_v\right)|\partial s \cap \partial t|$.
We prove this claim in Appendix~\ref{app:claim2} but briefly sketch it here for completeness.
The main ingredient of the proof is expanding $\sigma_B$ into a sum of terms depending on $B\cap \partial s$ for different services $s$.
We then lower bound this expansion and show $\sigma_B = \Omega(|A|)$.
On the other hand, the profitability condition implies that $\sigma_B \leq f_p(\pi, A) = O(|A|^{1/p})$.
Combining these inequalities yields the claim.
We note that if one adds extra constraints (\ie~$G$ is $d$-regular), then one can get achieve stronger bounds.

\subsection{Proof of Step 3}
The third step demonstrates that with sufficient incentives, one is able to rebalance enough stake to ensure that the $p$-attack $(A_2, B_2)$ is not viable.
Recall that a costly attack needs to be both profitable~\eqref{eq:costly-profit} and feasible~\eqref{eq:feasibility}.
When a set of node operators rebalances, they do not change the profitability of the $p$-attack $(A_2, B_2)$, but rather, make the setup infeasible.
Suppose $D_s \subset V$ is the set of node operators who rebalance into service $s$ after $p$-attack $(A_1, B_1)$.
Then the rebalance has successfully thwarted a $p$-attack $(A_2, B_2)$ if the following conditions hold:
\begin{align}
\text{[Attack 1 is profitable]} && f_p(\pi, A_1) &\geq \sigma_{B_1} \\
\text{[Attack 1 is feasible]} && \forall s \in A_1 \;\;\; \sigma_{\partial s \cap B_1} &\geq \alpha_s \sigma_{\partial s}\\
\text{[Attack 2 is profitable]} && f_p(\pi, A_2) & \geq \sigma_{B_2} \label{eq:attack2-profitable}\\
\text{[Attack 2 is feasible without rebalancing]} && \forall s \in A_2 \;\;\; \sigma_{(\partial s - B_1)\cap B_2} &\geq \alpha_s \sigma_{\partial s - B_1} \label{eq:feasible-no-rebalance}\\ 
\text{[Attack 2 is infeasible with rebalancing]} && \forall s \in A_2 \;\;\; \sigma_{(\partial s - B_1 + D_s) \cap B_2} &\leq \alpha_s \sigma_{\partial s - B_1 + D_s} \label{eq:infeasible-rebalance}
\end{align}
Our claim demonstrates that~\eqref{eq:infeasible-rebalance} holds under conditions on the rewards and overlap.
We first prove a simple lemma that will simplify the proof.
\begin{lemma}\label{lemma:d-lower-bound}
Suppose that $\frac{r_s}{c_s} \geq KS^{1/p} + 1$. Then there exists $\kappa > 0$ such that
\[
\sigma_{D_s(B_1)} \geq \kappa \sigma_{\partial s - B_1}
\]
\end{lemma}
\begin{proof}
  We first define the sets $D_s(B_1)$ as the set of node operators who are not profitable at $s$ prior to $B_1$ being attacked and are profitable afterwards.
    Formally, define
    \begin{align*}
    \mathcal{D}_s(B_1, r_s, c_s) = \left\{ D \subset V - B_1 : \forall v\in D,\;\sigma_v \geq \frac{c_s\sigma_{\partial s - B_1 + D}}{r_s-c_s}, \;
    \sigma_v \leq \frac{c_s\sigma_{\partial s}}{r_s-c_s}\right\}
    \end{align*}
    where we use eq.~\eqref{eq:profitable-impact} to define profitability conditions.
    Let $D_s(B_1) \in \argmax_{D \in \mathcal{D}_s(B_1, r_s, c_s)} \sigma_D$.
    Next, we bound the amount of stake in $D_s(B_1)$ by $\sigma_{\partial s}$ and the attack $\sigma_{B_1}$:
    \[
    \sigma_{D_s(B_1)} = \sum_{v \in D_s(B_1)} \sigma_v \geq \sum_{v \in D_s(B_1)} \frac{c_s \sigma_{\partial s - B_1 + D_s(B_1)}}{r_s - c_s} = |D_s(B_1)| \frac{c_s \sigma_{\partial s - B_1 + D_s(B_1)}}{r_s - c_s} 
    \]
    This implies that $\left(1 - |D_s(B_1)| \frac{c_s}{r_s - c_s}\right) \sigma_{D_s(B_1)} \geq \frac{|D_s(B_1)| c_s}{r_s - c_s} \sigma_{\partial s - B_1}$.
    Therefore when $|D_s(B_1)| \frac{c_s}{r_s-c_s} < 1$, this is a non-trivial bound: $\sigma_{D_s(B_1)} \geq \kappa \sigma_{\partial s - B_1}$.
    As $|D_s(B_1)| \leq KS^{1/p}$, this condition holds when $\frac{r_s}{c_s} \geq KS^{1/p} + 1$.
\end{proof}

\begin{claim}[Step 3]
    Suppose that $\frac{r_s}{c_s} \geq KS^{1/p} + 1$, $|A_t| = o(S)$ and for all $s \in S$ $\sigma_{\partial s} \geq (1+\left(\min_v \sigma_v\right)) K\left(\max_v \sigma_v\right)S^{1/p}$.
    Moreover, suppose that $\left(\min_s \alpha_s\right) \left(\max_v \sigma_v\right) \geq 2$.
    Then~\eqref{eq:infeasible-rebalance} holds.
\end{claim}
\begin{proof}
  From the feasibility of $(A_2, B_2)$ without rebalancing, eq.~\eqref{eq:feasible-no-rebalance}, we have $\alpha_s \sigma_{\partial s - B_1} \leq \sigma_{(\partial s - B_1) \cap B_2} \leq \sigma_{B_2}$.
  Since $|A| = o(S)$, the profitability of attack 2 (eq.~\eqref{eq:attack2-profitable}) implies
  \[
  \sigma_{B_2} \leq f_p(\pi, A_2) \leq \left(\max_s \pi_s\right) |A_2|^{1/p} \leq \left(\max_s \pi_s\right) S^{o(1/p)}
  \]
  This implies that $\sigma_{(\partial s - B_1)\cap B_2} = O(S^{o(1/p)})$.
  Similarly, we have $\sigma_{D_s(B_1)\cap B_2} \leq \sigma_{B_2} = O(S^{o(1/p)})$.
  From the lemma and by our assumption on $r_s$, we have $\sigma_{D_s(B_1)} \geq \kappa \sigma_{\partial s - B_1}$.
  Moreover, note that
  \[
  \sigma_{\partial s - B_1} \geq \sigma_{\partial s} - |B_1| \left(\max_v \sigma_v\right) \geq \left(\min_v \sigma_v\right) K \left(\max_v \sigma_v\right) S^{1/p} = \left(\max_s \pi_s\right)\left(\max_v \sigma_v\right) S^{1/p} 
  \]
  This implies that $\sigma_{\partial s - B_1 + D_s} \geq (1+\kappa)\sigma_{\partial s - B_1} \geq (1+\kappa)  \left(\max_s \pi_s\right) \left(\max_v \sigma_v\right) S^{1/p}$.
  Combined, we have
  \begin{align*}
  \sigma_{(\partial s - B_1 + D_s)\cap B_2} &= \sigma_{(\partial s - B_1)\cap B_2} + \sigma_{D_s \cap B_2} \leq 2 \left(\max_s \pi_s\right) S^{o(1/p)} \\
  &\leq (1+\kappa) \alpha_s \left(\max_s \pi_s\right) \left(\max_v \sigma_v\right) S^{1/p} \leq \alpha_s \sigma_{\partial s - B_1 + D_s}
  \end{align*}
  which proves~\eqref{eq:infeasible-rebalance}.
\end{proof}
We note that our assumptions are mild: the condition on $\sigma_{\partial s}$ simply says the initial stake at each service needs to be above a $p$ dependent threshold while the condition on $\min_s \alpha_s$ effectively says we need services that don't have low attack percentages relative to stake.

\subsection{Proof of Step 4}
As $\pi, \sigma$ are fixed, rebalancing can only change the feasiblity of an attack~\eqref{eq:feasibility}.
It is possible that a rebalance to thwart an attack $(A_2, B_2)$ makes a profitable yet infeasible $(\tilde{A}, \tilde{B}) \subset S \times V$ profitable and feasible.
If our rebalance set for service $s$ is $D_s \subset V$, this occurs when the following conditions hold:
\begin{align}
\text{[$(\tilde{A}, \tilde{B})$ is infeasible without rebalancing]} && \forall s \in \tilde{A} \;\;\; \sigma_{(\partial s - B_1)\cap \tilde{B}} &\leq \alpha_s \sigma_{\partial s - B_1} \label{eq:infeasible-no-rebalance}\\ 
\text{[$(\tilde{A}, \tilde{B})$  is feasible with rebalancing]} && \forall s \in \tilde{A} \;\;\; \sigma_{(\partial s - B_1 + D_s) \cap \tilde{B}} &\geq \alpha_s \sigma_{\partial s - B_1 + D_s} \label{eq:feasible-rebalance}
\end{align}
We demonstrate a simple sufficient condition depending on the rewards that ensures that there exists a rebalance where this does not occur.

\begin{claim}
    Suppose $\frac{r_s}{c_s} \geq 2\frac{\sigma_{\max}}{\sigma_{\min}}KS^{1/p} + 1$, then for all $D_s \in \mathcal{D}_s(B_1, r_s, c_s)$~\eqref{eq:feasible-rebalance} does not hold
\end{claim}
\begin{proof}
    Rewrite~\eqref{eq:feasible-rebalance} as 
    $\sigma_{\tilde{B} \cap D_s} \geq (\alpha_s \sigma_{\partial s - B_1} - \sigma_{(\partial s - B_1)\cap \tilde{B}}) + \alpha_s \sigma_{D_s}$.
    First note that
    $\alpha_s \sigma_{\partial s - B_1} - \sigma_{(\partial s - B_1) \cap \tilde{B}} \geq \alpha_s \sigma_{\partial s - B_1} - \sigma_{\tilde{B}} \geq \alpha_s \sigma_{\partial s - B_1} - \sigma_{\max} K S^{1/p}$.
    Any $D_s \in \mathcal{D}_s(B_1, r_s, c_s)$ has $\sigma_{\partial s} \geq \frac{r_s - c_s}{c_s} \sigma_v \geq 2\frac{\sigma_{\max}}{\sigma_{\min}}KS^{1/p} \sigma_v \geq 2\sigma_{\max} K S^{1/p}$ so $\alpha_s \sigma_{\partial s - B_1} - \sigma_{(\partial s - B_1) \cap \tilde{B}} \geq K\sigma_{\max}S^{1/p}$.
    Since any attacking coalition $\tilde{B}$ has $\sigma_{D_s \cap \tilde{B}} \leq \sigma_{\tilde{B}} \leq K\sigma_{\max} S^{1/p}$,~\eqref{eq:feasible-rebalance} does not hold.
\end{proof}
Therefore, with sufficient rewards, any attack that is infeasible prior to rebalancing will not be feasible after rebalancing.

\subsection{Proof of Steps 5 and 6}\label{sec:step5-6}
Provided that $\forall s \in S, \exists r_s > 0$ and $V \geq S$, it is clear that $\sigma_V \geq \left(\min_v \sigma_v \right) S$.
When combined with steps 1, 2, and 3, this implies equation~\eqref{eq:r-psi}.
What remains to be shown is that the assumptions that exist within step 2 --- namely that for all $s \in S, \utheta_s, \otheta_s = \Theta(S^{1/p})$ --- can also be incentivized via rewards $r_s$.

To do this, we consider operator-specific rewards $r_{sv}$ for $s \in S, v\in V$, which incentivize targeting a particular overlap range.
For any service $s$, let $\tau(s) = \argmax_{t \in S, t\neq s} \sigma_{\partial s \cap \partial t} \subset S$ and let $\partial \tau(s) = \{v \in V : \exists s\in \tau(s) \;\;v \in \partial s \}$.
For each service $s$, let $\delta_s > 0$ and define
\begin{equation}\label{eq:discriminating-rewards}
r_{sv}(r, \delta, \sigma) = (1-\delta_s \ones_{v \in \partial \tau(s)}) r_s
\end{equation}
This reward linearly decreases the reward received by $v$ if $v$ increases $\otheta_s$.
We make the following simple claim:
\begin{claim}[Step 5]\label{claim:otheta-ub}
    If $r_s = \Omega(S^{1/p})$ and $\delta_s \geq 1-\frac{S^{1/p}}{r_s}$, then $\otheta_s = O(S^{1/p})$
\end{claim}
\begin{proof}
    For these rewards, a node operator $v$ is profitable if $r_{sv} \frac{\sigma_v}{\sigma_{\partial s}} - c_s > 0$
    which implies the profitability with impact condition $\sigma_v >  \frac{c_s (\sigma_{\partial s} + \sigma_v)}{r_{sv}}$.
    This simplifies to $\sigma_v > \frac{c_s \sigma_{\partial s}}{r_{sv} - c_s}$.
    Now consider a $p$-attack $(A, B)$ and define the set
    \begin{align*}
    \mathcal{D}_s(B, \delta_s) = \left\{ D \subset V - B : \forall v\in D,\;\sigma_v \geq \frac{c_s\sigma_{\partial s - B + D}}{r_{sv}(\delta_s)-c_s}, \;
    \sigma_v \leq \frac{c_s\sigma_{\partial s}}{r_{sv}(\delta_s)-c_s}\right\}
    \end{align*}
    Let $D \in \mathcal{D}_s(B, \delta_s)$.
    By Lemma~\ref{lemma:d-lower-bound}, we have $\sigma_D \geq \kappa \sigma_{\partial s - B}$.
    This implies that
    \[
    \sigma_{\partial S - B + D} = \sigma_{\partial S - B} + \sigma_D \geq (1+\kappa)\sigma_{\partial s - B} \geq (1+\kappa)\left(\sigma_{\partial s} - \sigma_{\max} S^{1/p}\right)
    \]
    By the definition of $\mathcal{D}_s(B, \delta_s)$, we have
    \begin{align*}
    \sigma_D = \sum_{v \in D} \sigma_v &\geq c_s\sigma_{\partial s - B +D} \sum_{v \in D} \frac{1}{r_{sv}(\delta_s) - c_s} = (c_s\sigma_{\partial s - B +D}) \left(\frac{|D \cap \partial \tau(s)|}{r_s(1-\delta_s) - c_s} + \frac{|D - \partial \tau(s)|}{r_s - c_s} \right) \\
    &\geq (c_s\sigma_{\partial s - B +D})\frac{|D \cap \partial \tau(s)| + C |D-\partial \tau(s)|}{r_s(1-\delta_s) - c_s} \geq \frac{(1+\kappa) c_s C|D|}{r_s(1-\delta_s) - c_s}\left(\sigma_{\partial s} - \sigma_{\max} S^{1/p}\right) \\
    &\geq \frac{(1+\kappa) c_s K}{r_s(1-\delta_s) - c_s}\left(\otheta_s - \sigma_{\max} S^{1/p}\right)
    \end{align*}
    Note that the first inequality follows from the first bound in $\mathcal{D}_s(B, \delta_s)$, the second from the definition of $r_{sv}$.
    The third, fourth, and fifth inequalities stem from the fact that 
    $C = \frac{r_s(1-\delta_s) - c_s}{r_s - c_s} = \Omega\left(\frac{S^{1/p}}{r_s}\right)$
    and that $C|D| = C(|D\cap \partial \tau(s)| + |D-\partial\tau(s)|) \geq K$.
    Rearranging this yields
    \begin{equation}\label{eq:otheta-bound}
    \otheta_s \leq \frac{1}{(1+\kappa) c_s K} (r_s(1-\delta_s) - c_s + \sigma_{\max} S^{1/p})
    \end{equation}
    Given \eqref{eq:otheta-bound}, the choice of $\delta_s \geq 1 - \frac{S^{1/p}}{r_s}$ implies that $\otheta_s = O(S^{1/p})$.
\end{proof}
This claim shows that for sufficiently large discount $\delta_s$, one can ensure the maximum overlap satisfies the conditions of the prior steps.
We note a similar technique can be used for ensuring lower bounds on overlap that provides boosts or extra incentives for increasing overlap between services that are too low.

\section{Reward Selection Algorithms}\label{sec:algorithms}
We have demonstrated that provided rewards are sufficiently high, one can bound $R_{\psi}(G)$.
However, in practice, there is the question of how to compute optimal rewards $r^{\star}_s$ given a restaking graph $G$.
Given that the profit functions are submodular, we demonstrate that utilizing Algorithm~\ref{alg}, one can approximate the optimal rewards up to a multiplicative factor $E(S, p)$.

\paragraph{Sequential Profit Function.}
Consider the function $\mathsf{Profit}_p : C(G) \rightarrow \reals_+$ that is defined as:
\begin{equation}\label{eq:profit-function}
\mathsf{Profit}_p((A_1, B_1), \ldots, (A_T, B_T)) = \sum_{t=1}^T f_p(\pi, A_i) - \sigma_{B_i}
\end{equation}
This function is strictly positive for attack sequences.
Our goal will be to take an bound $T \in \naturals_+$ on sequence length and find a sequence $(\hat{A}_1, \hat{B}_1), \ldots, (\hat{A}_k, \hat{B}_k)$, $k \leq T$ such that $\mathsf{Profit}_p((\hat{A}_1, \hat{B}_1), \ldots, (\hat{A}_k, \hat{B}_k)) \geq \alpha(p) \max_{\substack{\mathcal{S} \in C(G), |\mathcal{S}| \leq T}} \mathsf{Profit}_p(\mathcal{S})$.
Given this sequence, one can solve for rewards that ensure that attack $(\hat{A}_2, \hat{B}_2)$ is invalid in the presence of rebalancing.

\paragraph{Optimal Rewards.}
Let $(A^{\star}_1, B^{\star}_1), \ldots, (A^{\star}_T, B^{\star}_T) \in \argmax_{\mathcal{S} \in C(G)} \mathsf{Profit}(\mathcal{S})$, then we can compute the optimal rewards $r^{\star}_s$ to ensure that $R_{\psi}(G)$ by computing the minimum rewards such that~\eqref{eq:infeasible-rebalance} holds.
This implies that we need $\alpha_s \sigma_{\partial s - B^{\star}_1 - D(r_s)} \geq \sigma_{(\partial s - B^{\star}_1 + D_s) \cap B^{\star}_2}$ for all $s \in A^{\star}_2$.
This can be formulated as the following minimax program:
\begin{align}
f(r_s) &= \max_{D \in \mathcal{D}_s(B_1, r_s, c_s)} \sum_{s \in A^{\star}_2} \log(\alpha_s \sigma_{\partial s - B^{\star}_1 + D} - \sigma_{(\partial s - B^{\star}_1 + D) \cap B^{\star}_2}) \nonumber \\
r^{\star}_s &= \min\{ r_s \geq 0 : f(r_s) > 0\} \label{eq:r-star}
\end{align}
We consider a convex relaxation of this problem.
Note that $D \in \mathcal{D}_s(B_1, r_s, c_s)$ implies a $V$-size set of linear constraints on $D$ of the form $\sigma_v \geq c_s \sigma_{\partial s - B^{\star}_1 + D}{r_s-c_s} \iff  \sigma_D \leq (r_s-c_s)\sigma_v - \sigma_{\partial s - B^{\star}_1}$
since $\partial s, B^{\star}_1$ are fixed.
Let $\Delta^n = \{ x \in \reals_+^n : \sum_i x_i \leq 1\}$ be the simplex.
We define
\begin{equation}\label{eq:f-rs-k}
\hat{f}(r_s, k) = \max_{D \in \left((r_s-c_s)k - \sigma_{\partial s - B_1^*}\right) \Delta^n} \sum_{s \in A^{\star}_2} \log(\alpha_s \sigma_{\partial s - B^{\star}_1 + D} - \sigma_{(\partial s - B^{\star}_1 + D) \cap B^{\star}_2})
\end{equation}
Note that $\hat{f}$ is monotone increasing in $r_s$ and $k$ and that this is a logarithmic barrier problem over a convex set~\cite[Ch. 11]{boyd2004convex}.
Moreover, note that $\hat{f} \geq f$ since $\mathcal{D}_s(B_1^{\star}, r_s, c_s) \subset \left((r_s-c_s)(\max_{v}\sigma_v) - \sigma_{\partial s - B_1^*}\right) \Delta^n$. 
This implies that unless $f(r_s) < 0$ for all $r_s > 0$, then we can find $r_s$ by a combination of bisection and solving~\eqref{eq:f-rs-k} by an interior point method.

\paragraph{Approximation Guarantees.}
We prove the following claim in Appendix~\ref{app:algo}:
\begin{claim}
    Suppose~\eqref{eq:r-star} is feasible for some $r_s > 0$ and that $\frac{\min_s \pi_s}{\max_s \pi_s} \geq C$.
    Then Algorithm~\ref{alg} returns an approximation $\hat{r}_s$ that satisfies
    \[
    \hat{r}_s \geq \frac{C r_s^*}{S^{1/p}} + G
    \]
    where $G$ is the integrality gap $G = \sup_{r_s > 0} \max_{v \in V} \hat{f}(r_s, \sigma_v) - f(r_s)$.
\end{claim}
\section{Conclusion and Future Directions}
Our results demonstrate that cascade risk in restaking networks can be bounded with incentives.
Services can utilize our results to choose a rewards $r_s$ that ensure they can only be attacked by $\Omega(S^{1/p})$-sized attacking coalitions.
This allows smaller services to choose a threat model that offers fewer rewards should their local corruption profit be low.
On the other hand, our results show that strategic node operators (such as liquid restaking tokens~\cite{walter-avs}, who aggregate stake and delegate to node operators) need to rebalance efficiently to ensure network security.
Future work includes numerical evaluation of our algorithm on live restaking networks~\cite{u1-restaking} and to account for the price volatility and liquidity of $r_s$.
The latter problem exists as services often pay rewards in their native tokens rather than in the restaked asset~\cite{llama-risk-rewards, eigenlayer-rewards, symbiotic-rewards} and have to consider the impact of price and liquidity on whether their rewards are sufficient for node operators to be profitable.

\section{Acknowledgments}
We want to thank Tim Roughgarden, Naveen Durvasula, Soubhik Deb, Sreeram Kannan, Victor Xu, Walter Li, Gaussian Process, Manvir Schneider, Theo Diamandis, Matheus Ferreria, Guillermo Angeris, Kshitij Kulkarni for helpful comments and inspiration.

\printbibliography
\appendix


\section{Unions of $p$-attacks need not be valid attacks}\label{app:unions}
We first demonstrate an explicit example of a sequence $(A_1, B_1), \ldots, (A_T, B_T)$ such that for $f_{\infty}(\pi, A)$, we have
\begin{align*}
&&f_{\infty}(\pi, \cup_t A_t) &\leq \sigma_{\cup_{t} B_t} \\ 
\forall t \in [T]&& f_{\infty}(\pi, A_t) &\geq \sigma_{B_t}
\end{align*}
This will provide intuition for the example we show for general $p$.

Suppose that we have $\max_{s\in A_t} \pi_s = 1.1$ for all $t \in [T]$ and $\sigma_{B_t} = 1$ for all $t$.
Then we have
\[
\max_{s \in \cup_{t} A_t} \pi_s = \max_{t\in [T]} \max_{s \in A_t} \pi_s = 1.1
\]
On the other hand, we have $\sigma_{\cup_t B_t} = \sum_{t=1}^T \sigma_{B_t} = T$.
Therefore we have
\[
f_{\infty}(\pi, A_t) - \sigma_{B_t} = \max_{s \in A_t} \pi_s - \sigma_{B_t} = 1.1 - 1 = 0.1
\]
and
\[
f_{\infty}(\pi, \cup_t A_t) - \sigma_{\cup_t B_t} = 1.1 - T \leq 0 
\]
for $T > 1$.
We claim that a sufficient condition for the union of valid $p$-attacks to not be a valid $p$-attack is 
\[
\sum_{t\in[T]} |A_t|^{1/p} < \frac{\left(\min_v \sigma_v\right) \sum_{t\in[T]}|B_t|}{\max_{s \in \cup_t A_t} \pi_s}
\]
To see this, note that $f_p(\pi, A) \leq \left(\max_{s\in A} \pi_s\right) |A|^{1/p}$ and when this sufficient condition holds, this implies that
\[
f_p(\pi, \cup_t A_t) \leq \left(\max_{s\in \cup_t A_t} \pi_s\right) \sum_{t\in[T]}|A_t|^{1/p} \leq \left(\min_v \sigma_v\right) \sum_{t \in [T]}|B_t| \leq \sigma_{\cup_t B_t}
\]
From this, one can see that is is relatively easy to modify our example for $f_{\infty}$ to construct examples of $p$-attacks that are invalid.

\section{Overlap is easier to control than $\gamma$-security}\label{app:overlap}
We construct an example to show that the overlap $\theta_{s, t}$ is a finer means to control cascading than $\gamma$-security via an example.
Consider the restaking graph in Figure~\ref{fig:overlap} where $\theta_{1,2} = \theta_{2,1} = \sigma_{\cap}$, as the services only overlap in validator $v_{\cap}$.
Suppose that the potential attacks $(\{s_i\}, \{v_i\})$ are profitable but infeasible, \ie~
\begin{align*}
    \pi_i > \sigma_i && \sigma_i < \alpha_s (\sigma_i + \sigma_{\cap})
\end{align*}
Similarly, the potential attacks $(\{s_i\}, \{v_{\cap}\})$ are profitable and infeasible,
\ie~
\begin{align*}
\pi_i > \sigma_{\cap} && \sigma_{\cap} < \alpha_s (\sigma_i + \sigma_{\cap})
\end{align*}
On the other hand, suppose that the potential attacks $(\{s_i\}, \{v_i, v_{\cap}\})$ are unprofitable but are (tautologously) feasible, \ie~
\begin{align*}
\pi_i < \sigma_i + \sigma_{\cap} && \sigma_i + \sigma_{\cap} > \alpha_s (\sigma_i + \sigma_{\cap})
\end{align*}
Finally, suppose that the potential attack $(\{s_1, s_2\}, \{v_1, v_2, v_{\cap}\})$ has zero profit and is (tautologously) feasible, \ie~
\begin{align*}
    \pi_1 + \pi_2 = \sigma_1 + \sigma_2 + \sigma_{\cap} && \sigma_i + \sigma_{\cap} > \alpha_s (\sigma_i + \sigma_{\cap})
\end{align*}
\begin{figure}
\begin{center}
\begin{tikzpicture}

  \tikzstyle{left node}=[circle, draw, fill=orange, text=white, font=\bfseries, minimum size=20pt, inner sep=0pt]

  \tikzstyle{right node}=[circle, draw, fill=blue, text=white, font=\bfseries, minimum size=20pt, inner sep=0pt]

  \node at (-1,2.25) {$\pi_1$};
  \node at (-1,0.75) {$\pi_2$};

  \node at (-1,4) {$\pi$};

  \node[left node] (L1) at (0,2.25) {$s_1$};
  \node[left node] (L2) at (0,0.75) {$s_2$};

  \node[right node] (R1) at (2,3) {$v_1$};
  \node[right node] (R2) at (2,1.5) {$v_{\cap}$};
  \node[right node] (R3) at (2,0) {$v_2$};

  \node at (3,3) {$\sigma_1$};
  \node at (3,1.5) {$\sigma_{\cap}$};
  \node at (3,0) {$\sigma_2$};

  \node at (3,4) {$\sigma$};

  \draw (L1) -- (R1);
  \draw (L2) -- (R3);
  \draw (L1) -- (R2);
  \draw (L2) -- (R2);

\end{tikzpicture}
\caption{Example of a graph with overlap where the overlap controls the cascading likelihood}
\label{fig:overlap}
\end{center}
\end{figure}
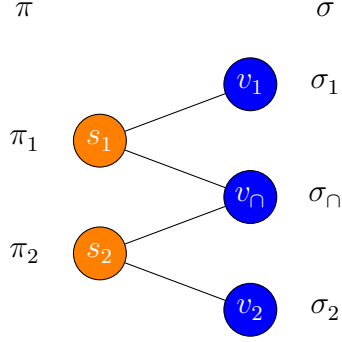
Now suppose let $\sigma_1 = \sigma_2 = \sigma$ and that $\sigma_{\cap} = K \sigma$ for $K \geq 1$.
When $\psi < \frac{\sigma_{\cap}}{\sigma_V}$, we have
\[
R_{\psi}(G) < \psi + \frac{\sigma_i}{\sigma_V} < \frac{1}{2+K}
\]
since $\sigma_V = \sigma_1 + \sigma_2 + \sigma_{\cap} = (K+2)\sigma$.
This is because no attack is feasible without removing $\sigma_{\cap}$ and losing $\sigma_i$ for $i\in\{1,2\}$ is not sufficient to cause a cascade.
On the other hand, note that
\[
\frac{\pi_i - \sigma}{(K+2)\sigma} < \frac{\sigma_{\cap}}{\sigma_V} = \frac{K}{K+2} < \frac{\pi_i}{(K+2)\sigma}
\]
For large $K$ and $\pi_i \leq (1-\epsilon)(K+2)\sigma$, this implies that $R_{\psi}(G) < \frac{1}{2+K}$ for $\psi < 1-\epsilon$.
Under these conditions, the graph has low cascading risk and it is completely controlled by $\sigma_{\cap}$.
This implies that local changes to $\sigma_{\cap}$ alone (without needing to adjust $\sigma_i$) are sufficient to ensure low cascade risk.
As we demonstrate in the rebalancing section of~\S\ref{sec:model}, in the presence of rebalancing, one can bound how low $\sigma_{\cap}$ can go and thus, one can bound the worst case cascade.

On the other hand, we will now compute how large $\sigma_{\cap}$ would have to be if we enforced $\gamma$-security.
Let $\epsilon_i = \pi_i - \sigma$ be the profit achieved if $v_i$ could attack $s_i$.
From the second condition, we have $\sigma_{\cap} > \pi_i - \sigma_i = \epsilon_i$.
Suppose we consider a $\gamma$-secure version of this graph where we have stakes $\sigma_1, \sigma_2, \sigma_{\cap}^{\gamma}$.
The $\gamma$-secure condition implies that
$(1+\gamma)(\pi_1 + \pi_2) < \sigma_1 + \sigma_2 + \sigma_{\cap}^{\gamma}$ which can be rearranged to
\[
\sigma_{\cap}^{\gamma} > (\pi_1 - \sigma_1) + (\pi_2 - \sigma_2) + \gamma(\pi_1 + \pi_2) = \epsilon_1 + \epsilon_2 + \gamma(\pi_1 + \pi_2) > 2\sigma_{\cap} + \gamma(\pi_1 + \pi_2)
\]
This implies that $\sigma_{\cap}^{\gamma} - \sigma_{\cap} = \Omega(\pi_1 + \pi_2)$,
suggesting that $\gamma$-security is far too strong as all services in this graph have to attract $\Omega(\pi_1 + \pi_2)$ stake.
Concretely, suppose that $\pi_1 = \$1,000$ and $\pi_2 = \$1,000,000,000$.
The $\gamma$-security condition implies that service 1, which can only be exploited for at most \$1,000 has to attract stake proportional to \$1,000,001,000 in order to be $\gamma$-secure.
Yet, if one can control only the stake held by $v_{\cap}$, one can avoid most cascades (\ie~for $\psi < \frac{\sigma_{\cap}}{\sigma_V}$) without requiring such a high amount of stake.

\section{Proof of Claim 2}\label{app:claim2}

\begin{proof}
    Let $V_s = B  \cap \partial s$ and note that $B = \bigcup_{s\in V} V_s$.
    By inclusion-exclusion, we have
    \[
    \sigma_B = \sum_{s \in S} \sigma_{V_s} - \sum_{1 \leq s < s' \leq S} \sigma_{V_s \cap V_{s'}} + \ldots = \sum_{\emptyset \neq J \subset S} (-1)^{|J|} \sigma_{\cap_{s \in J} V_s} 
    \]
    We now claim that there exists $\xi \in (0,1)$ such that $\sigma_{V_{s_1}} \cap \cdots \cap \sigma_{V_{s_k}} \leq \xi^{k-1} \sigma_{V_{s_j}}$ for any $j \in [k]$.
    By assumption, every intersection $\partial s \cap \partial t$ has at least $\delta KS^{1/p}$ more node operators than any attacking coalition $B \subset V$, since $|B| \leq KS^{1/p}$.
    This means that for any pair $s, t \in S$, we have $|\partial s \cap \partial t - B| \geq \delta KS^{1/p}$.
    Therefore, we have
    \[
    \sigma_{V_s \cap V_t} = \sigma_{(B\cap \partial s) \cap (B\cap \partial t)} \leq \left(1 - \frac{\left(\min_v \sigma_v\right) | \partial s \cap \partial t - B|}{\sigma_{\partial s \cap \partial t}}\right) \sigma_{V_s} \leq \left(1 - \frac{\delta \left(\min_v \sigma_v\right) K S^{1/p}}{\otheta_s}\right) \sigma_{V_s}
    \]
    Since $\otheta_s \leq \left(\max_v \sigma_v\right) KS^{1/p}$, we have
    $\sigma_{V_s \cap V_t} \leq \left(1 - \frac{\delta \left(\min_v \sigma_v\right)}{ \left(\max_v \sigma_v\right)}\right) \sigma_{V_s} = \xi \sigma_{V_s}$,
    where we defined $\xi = 1 - \frac{\delta \left(\min_v \sigma_v\right)}{\left(\max_v \sigma_v\right)}$.
    One can recurse this argument to get $\sigma_{V_{s_1} \cap V_{s_2} \cdots \cap V_{s_k}} \leq \xi^{k-1} \sigma_{V_{s_i}}$.
    Therefore, we have
    \begin{align*}
    \sigma_B  &= \sum_{s \in S} \sigma_{V_s} - \sum_{1 \leq s < s' \leq S} \sigma_{V_s \cap V_{s'}} + \ldots = \sum_{\emptyset \neq J \subset S} (-1)^{|J|} \sigma_{\cap_{s \in J} V_s} \\
              &\geq \sum_{s\in S} \sigma_{V_s} - (S-1) \xi \sigma_{V_s} + \sum_{s, s', s''} \sigma_{V_s \cap  V_{s'} \cap V_{s''}} - \binom{S-1}{3} \xi^3 \sigma_{V_s}   - \cdots \\
              &\geq \sum_{s\in S}\sigma_{V_s} \left(1 - \sum_{i=0}^{\lfloor \frac{S-1}{2} \rfloor} \binom{S-1}{2i + 1} \xi^{2i+1} \right) \geq \left(1 - \sum_{i=0}^{\lfloor \frac{S-1}{2} \rfloor} \left(\frac{e (S-1)}{2i + 1}\right)^{2i+1} \xi^{2i+1} \right) \sum_{s\in S} \sigma_{V_s}
    \end{align*}
    where the last step uses the binomial bound $\binom{n}{k} \leq \left(\frac{en}{k}\right)^k$.
    Note that $\xi = 1 - \frac{\delta\left(\min_v \sigma_v\right)}{\left(\max_v \sigma_v\right)} \leq \frac{(S-1)^{2/p}}{e(S-1)^2}$, by assumption. This implies that $\sum_{i=0}^{\lfloor \frac{S-1}{2} \rfloor} \left(\frac{e (S-1)}{2i + 1}\right)^{2i+1} \xi^{2i+1} \leq \frac{C}{S-1}$ for a constant $C < 1$.
This implies for a costly attack with profit function $f_p$ that we have
\[
\left(1-\frac{C}{S-1}\right)\left(\min_v \sigma_v\right) |A| \leq \sigma_B \leq f_p(\pi, A) \leq \left(\max_{s} \pi_s\right) |A|^{1/p}
\]
which implies the claim $|A|^{1-1/p} \leq \frac{\max_s \pi_s}{\min_v \sigma_v} \frac{2}{\left(1-\frac{C}{S-1}\right)} = \frac{\max_s \pi_s}{\min_v \sigma_v} \frac{2(S-1)}{S-1-C} \leq  \frac{\max_s \pi_s}{\min_v \sigma_v} \frac{2(S-1)}{S-2}$
\end{proof}

\section{Approximation Algorithm}\label{app:algo}
\begin{algorithm}
\caption{Compute Approximately Optimal Rewards}\label{alg}
\SetKwInOut{Input}{input}\SetKwInOut{Output}{output}

\Input{Incentivized restaking graph: $G = (S, V, E, \alpha, \sigma, \pi, r, c)$ \\
Upper bound on maximal attack sequence length: $T \in \naturals$ \\
Curvature of submodularity: $p \in (1, \infty)$
}
\Output{Approximately Optimal Reward Vector: $\hat{r} \in \reals^S_+$}
$(\hat{A}_1, \hat{B}_1), \ldots, (\hat{A}_k, \hat{B}_k) \leftarrow \mathsf{Greedy}(\mathsf{Profit}_p, T)$ (See Algorithm \ref{alg2})\;
$r \leftarrow \emptyset$\;
\For {$s \in S$}{
$\mathcal{D} \leftarrow \emptyset$\;
$\mathsf{maxSoFar} \leftarrow -\infty$ \;
$r_s \leftarrow \frac{\sigma_{\partial s - \hat{B}_1}}{\max_{v \in V} \sigma_v} + c_s$\;
\While{$\mathsf{maxSoFar} < 0$} { 
\For {$v \in V$} {
     // Solve for $\hat{f}$ with an interior point method\\
     $\mathsf{maxSoFar} \leftarrow \max(\hat{f}(r_s, \sigma_v), \mathsf{maxSoFar})$\; 
}
\If{$\mathsf{maxSoFar} \geq 0$} {
break\;
}
$r_s \leftarrow 2r_s$\;
}
}
\For{$r_s \in r$} {
// Increase reward by profit approximation error \\
$r_s \leftarrow \frac{r_s}{E(p)}$\;
}
\Return $r$\;
\end{algorithm}
Algorithm~\ref{alg} proceeds to output a set of rewards $r_s$ using the following steps:
\begin{enumerate}
\item Let $C(G) = \{(A_1, B_1), \ldots, (A_T, B_T) : (A_i, B_i) \in (S-\cup_{j=1}^{i-1} A_i, V-\cup_{j=1}^{i-1} B_i)\}$ be the set of possible attack sequences on a restaking graph $G$ and let $C_T(G) = \{ s \in C(G) : |s| \leq T\}$ be the set of sequences of length at most $T$.
\item For a sequence $(A_1, B_1), \ldots, (A_T, B_T) \in C_T(G)$, define the net profit function
\begin{equation}\label{eq:profit-function}
\mathsf{Profit}_p((A_1, B_1), \ldots, (A_T, B_T)) = \sum_{t=1}^T f_p(\pi, A_i) - \sigma_{B_i}
\end{equation}
\item Given an upper bound $T \in \naturals$ on the attack length, utilize a greedy algorithm to find $(\hat{A}_1, \hat{B}_1), \ldots, (\hat{A}_k, \hat{B}_k) \in C_{T}(G)$, $k \in [T]$, such that
\begin{equation}\label{eq:profit-poa}
\mathsf{Profit}((\hat{A}_1, \hat{B}_1), \ldots, (\hat{A}_k, \hat{B}_k)) \geq \alpha(p) \max_{s \in C_T(G)} \mathsf{Profit}(s)
\end{equation}
where the approximation factor $\alpha(p)$ only depends on the choice of $p$ and bounds on $\min_s \pi_s, \max_s \pi_s$
\item Given this sequence, compute a rebalance $D$ and rewards $\hat{r}_s$ that ensure that $(\hat{A}_2, \hat{B}_2)$ is infeasible after rebalancing
\item Return rewards $r_{sv} = \frac{\hat{r}_{sv}}{1-e^{-\alpha(p)}}$
\end{enumerate}
We claim such rewards will incentive a graph $G$ with $R_{\psi}(G) \leq \psi + \frac{C}{S^{1-1/p}}$ while only being $G\alpha(p)$-times worse than the minimum rewards needed to achieve such a graph in the worse case, where $G$ is the integrality gap.

\paragraph{Sequential Submodularity.}
We first claim that the profit function~\eqref{eq:profit-function} is a \emph{sequentially submodular} function~\cite{alaei2021maximizing, balkanski2018adaptive}.
To define a sequentially submodular function, we consider sets of sequences $ \mathcal{S}^{\infty} = \bigcup_{k\in\naturals} \{(s_1, \ldots, s_k) : s \in \mathcal{S} \}$ of a set $\mathcal{S}$.
We define a partial order $<$ for $A, B \in \mathcal{S}^{\infty}$ where $A < B$ iff $A$ is a subsequence of $B$.
Moreover, given two sequences $A = (s_1, \ldots, s_k), B = (t_1, \ldots, t_j) \in \mathcal{S}^{\infty}$, we define the concatenation $A \perp B = (s_1, \ldots, s_k, t_1, \ldots, t_j) \in \mathcal{S}^{\infty}$. 
A function $f : \mathcal{S}^{\infty} \rightarrow \reals$ is sequentially submodular if for all $A < B$ and $C \in \mathcal{S}^{\infty}$ we have
\[
f(B \perp C) - f(B) \leq f(A \perp C) - f(A)
\]
Note that this is a diminishing marginal utility condition, much like the standard set submodularity definition, $f(S\cup T) + f(S\cap T) \leq f(S) + f(T)$.
In our scenario, we have $\mathcal{S}^{\infty} = C(G)$ and note that $\mathsf{Profit}$ is a sequentially submodular function since each term $f_p(\pi, A) - \sigma_B$ is set submodular (as it is the sum of two submodular functions).-

\begin{algorithm}
\caption{Greedy Sequential Submodular Optimization (Algorithm 3 of~\cite{alaei2021maximizing})}\label{alg2}
\SetKwInOut{Input}{input}\SetKwInOut{Output}{output}

\Input{Sequentially submodular function $f : S^{\infty} \rightarrow \reals$, \\
Time horizon $T \in \naturals$
}
\Output{$(A_1, B_1), \ldots, (A_k, B_k)$ with $k < T$}
$(\hat{A}_1, \hat{B}_1), \ldots, (\hat{A}_k, \hat{B}_k) \leftarrow \mathsf{Greedy}(\mathsf{Profit}_p, T)$\;
$n \leftarrow$ length of $S$\;
Initialize $t \leftarrow 0; i \leftarrow 1; H \leftarrow \emptyset$;
\While{$t <T$}{
    Find $s_i \in \mathcal{S}$ such that $u(H \perp s_i) - u(H) \geq \alpha \max_{s \in S} u(H \perp s) - u(H)$ \\
    $H \leftarrow H \perp s_i$
}
\Return $H$\;
\end{algorithm}

\paragraph{Approximation Error.}
We next recall a theorem from~\cite{alaei2021maximizing} that demonstrates that a greedy algorithm has low approximation error for sequentially submodular functions:
\begin{theorem}{\cite[Thm. 3]{alaei2021maximizing}}\label{thm:alaei}
Suppose $f : \mathcal{S}^{\infty} \rightarrow \reals$ is a sequentially submodular function.
Suppose that for a history $S = (s_1, \ldots, s_T) \in \mathcal{S}^{\infty} $that $f$ satisfies the following for all $t \in [T-1]$
\[
f((s_1, \ldots, s_t) \perp s_{t+1}) - f(s_1, \ldots, s_t) \geq \alpha \max_{s \in \mathcal{S}} f((s_1, \ldots, s_t) \perp s) - f((s_1, \ldots, s_t))
\]
for $\alpha > 0$.
Then we have
\[
\frac{f((s_1, \ldots, s_T))}{\max_{s \in \mathcal{S^{\infty}}} f(s)} \geq 1 - \frac{1}{e^{\alpha}}
\]
where $(s_1, \ldots, s_T)$ is the history generated by the greedy algorithm~\ref{alg2}
\end{theorem}
\noindent This theorem implies that if we can find a lower bound on the marginal increase, then we can bound the worst case approximation error.
We now claim that if $\frac{\min_s \pi_s}{\max_s \pi_s} \geq K$ and as $|A_t| \geq 1$, then~\eqref{eq:profit-poa} holds with 
\[
\alpha(p) \geq KS^{-\frac{1}{p}}
\]
To see this, note that
\[
f_p(\pi, A_i)^p = \sum_{s \in A_i} \pi_s^p \geq (\min_s \pi_s)^p |A_i| \geq K^p (\max_s \pi_s)^p \geq \frac{K^p (\max_s \pi_s)^p}{S} S \geq \frac{K^p}{S} \sum_{s\in A}\pi_s^p
\]
for any other set $A$.
This implies that
\[
f_p(\pi, A_i) \geq \frac{K}{S^{1/p}} \max_A f_p(\pi, A)
\]
which implies that $\alpha(p) \geq \frac{K}{S^{1/p}}$

Generally, we will have $KS^{-1/p} \leq \sigma_{\max} V$.
This implies that $\alpha(p)$ increases as $p$ increases.
Using Theorem~\ref{thm:alaei}, this implies that the greedy algorithm of~\cite{alaei2021maximizing} achieves an approximation error of
\[
E(p) = 1-e^{-\alpha(p)} = 1-e^{-KS^{-1/p}} \geq \frac{K}{2S^{1/p}}
\]
where the last inequality holds when $K/S^{1/p} \leq \frac{1}{2}$
This error is highest when $p = 1$ and lowest when $p = \infty$.
As such, a service can view choosing $p \in [1, \infty)$ as choosing a security level (\eg~secure up to node operator cartels of size $O(S^{1/p})$) and then utilize our algorithm to choose the rewards to ensure security.
These rewards are guaranteed to be within $O(\frac{1}{E(p)})$ of the minimum possible rewards needed to achieve security but can be easily computed.

Since Algorithm~\ref{alg} optimizes $\hat{f}(r_s, k)$ instead of $f(r_s)$, it computes an overestimate of the optimal rewards if $f(r_s)$ is feasible.
This overestimate is bounded by the integrality gap $G =\max_{r_s > 0} \max_{v\in V} f(r_s, \sigma_v) - f(r_s)$, which gives the additive term in the approximation error.
\end{document}